\documentclass[journal]{IEEEtran}
\usepackage{smartdiagram}
\usepackage[noadjust]{cite}
\usepackage{mathrsfs}
\usepackage{mathtools}
\usepackage{graphicx}
\usepackage{textcomp}
\usepackage{xcolor}
\usepackage{booktabs}
\usepackage[retain-explicit-plus,qualifier-mode=text]{siunitx}
\usepackage{soul}
\usepackage{xcolor}
\usepackage{comment}
\usepackage[printonlyused,nohyperlinks, nolist]{acronym}
\usepackage[shortcuts,acronym,automake=false]{glossaries}
\usepackage{multirow}

\usepackage{enumitem}

\ifCLASSINFOpdf
\else
\fi
\usepackage{amsmath}
\usepackage{amssymb,amsfonts}
\ifCLASSOPTIONcompsoc
  \usepackage[caption=false,font=normalsize,labelfont=sf,textfont=sf]{subfig}
\else
  \usepackage[caption=false,font=footnotesize]{subfig}
\fi

\usepackage{stfloats}
\usepackage{url}

\usepackage{multirow}

\usepackage{amsthm}
\usepackage[hidelinks,bookmarks=false]{hyperref}

\usepackage[printonlyused,nohyperlinks,nolist]{acronym}
\usepackage{algorithm}
\usepackage{algpseudocode}

\usepackage[capitalise]{cleveref}
\crefname{enumi}{condition}{conditions}
\crefname{problem}{Problem}{Problems}
\crefformat{problem}{Problem~(#2#1#3)}
\Crefformat{problem}{Problem~(#2#1#3)}
\newcommand{\crefrangeconjunction}{--}

\usepackage{bookmark}

\usepackage{upgreek}
\usepackage{dsfont}

\usepackage[range-phrase=--]{siunitx}
\DeclareSIUnit{\belmilliwatt}{Bm}
\DeclareSIUnit{\dBm}{\deci\belmilliwatt}
\usepackage{import}
\usepackage{color}

\usepackage{pgf}

\DeclareUnicodeCharacter{2212}{-}

\theoremstyle{plain}
\newtheorem{lemma}{Lemma}
\newtheorem{proposition}{Proposition}
\newtheorem{corollary}{Corollary}

\theoremstyle{definition}
\newtheorem{definition}{Definition}

\theoremstyle{remark}
\newtheorem{remark}{Remark}

\newcommand{\eqrefrange}[2]{\eqref{#1}\crefrangeconjunction\eqref{#2}}

\newcommand{\Norm}[2][\empty]{\left\|{#2}\right\|_{#1}}
\newcommand{\norm}[2][\empty]{\left|{#2}\right|_{#1}}
\newcommand{\fro}[1]{\Norm[\mathrm{F}]{#1}}

\newcommand{\infnorm}[1]{\norm[\infty]{#1}}

\def\a{{\mathbf a}}
\def\b{{\mathbf b}}
\def\c{{\mathbf c}}

\def\s{{\mathbf s}}

\def\e{{\mathbf e}}

\def\v{{\mathbf v}}
\def\x{{\mathbf x}}
\def\y{{\mathbf y}}
\def\z{{\mathbf z}}
\def\w{{\mathbf w}}
\def\q{{\mathbf q}}
\def\r{{\mathbf r}}

\def\yH{{\hat{\y}}}
\def\yh{{\hat{y}}}

\def\ao{\a_{\circ}}
\def\bo{\b_{\circ}}
\def\co{\c_{\circ}}
\def\wo{\w_{\circ}}

\def\yfs{y_{\mathrm{fs}}}

\def\A{{\mathbf A}}
\def\B{{\mathbf B}}
\def\D{{\mathbf D}}

\def\G{{\mathbf G}}
\def\H{{\mathbf H}}
\def\Hr{{\H_{\mathrm r}}}
\def\R{{\mathbf R}}
\def\S{{\mathbf S}}
\def\M{{\mathbf M}}
\def\X{{\mathbf X}}
\def\U{{\mathbf U}}
\def\V{{\mathbf V}}
\def\C{{\mathbf C}}

\def\Z{{\mathbf Z}}

\def\Gt{\G_{\mathrm{tx}}}
\def\Gr{\G_{\mathrm{rx}}}

\def\mf{\mathfrak{m}}
\def\mgg{\widehat{\mf}_{\S}}
\def\mS{\mf_{\S}}
\def\mSo{\mf_{\S_1}}

\def\Tc{T_{\mathrm{c}}}

\def\Q{{\mathbf Q}}
\def\I{{\mathbf I}}

\def\W{{\mathbf W}}
\def\Z{{\mathbf Z}}
\def\Ts{T_{\mathrm S}}

\def\Pt{P_{\mathrm{tx}}}
\def\Psat{P_{\mathrm{sat}}}

\def\CN{\mathcal{CN}}

\def\RR{{\mathbb R}}

\def\CC{{\mathbb C}}
\def\SS{{\mathbb S}}
\def\EE{{\mathbb E}}

\def\Pr{{\mathrm{Pr}}}

\def\trace{\mathrm{tr}}

\def\minimize{\mathrm{minimize}}
\def\maximize{\mathrm{maximize}}

\def\st{\mathrm{s.t.}}
\def\rank{\mathrm{rank}}

\def\brx{\bar{\varrho}_\x}

\def\setC{\mathcal{C}}

\def\setL{\mathcal{L}}

\def\setP{\mathcal{P}}

\def\setW{\mathcal{W}}
\def\setU{\mathcal{U}}

\def\setZ{\mathcal{Z}}

\def\Wd{\setW_{\triangleright}}
\def\Cd{\setC_{\triangleright}}
\def\Wp{\setW_{\oslash}}
\def\Cp{\setC_{\oslash}}
\def\Wl{\setW^{L'}}
\def\Ck{\setC^{K'}}

\def\Eig{\boldsymbol{\Sigma}}
\def\Gam{\boldsymbol{\Gamma}}

\def\Zero{\boldsymbol{0}}

\def\scb{\sigma_{\mathrm{cb}}}
\def\stx{\sigma_{\mathrm{tx}}}
\def\srx{\sigma_{\mathrm{rx}}}

\def\h{\mathsf{H}}
\def\t{\mathsf{T}}

\def\hsi{h_{\mathrm{si}}}
\def\hr{h_{\mathrm{r}}}
\def\fc{f_{\mathrm{c}}}

\def\zsol{\z_\star}

\def\sgT{\sigma_{\mathrm{T}}}
\def\sgQ{\sigma_{\mathrm{Q}}}
\def\bQ{\mathfrak{b}_{\mathrm{Q}}}

\def\ysi{y_{\mathrm{si}}}
\def\yr{y_{\mathrm{r}}}

\begin{document}

\newacronym{isac}{ISAC}{integrated sensing and communication}
\newacronym{bler}{BLER}{block error rate}
\newacronym{tx}{TX}{transmission}
\newacronym{rx}{RX}{reception}
\newacronym{mse}{MSE}{mean squared error}
\newacronym{nr}{NR}{new radio}
\newacronym{ula}{ULA}{uniform linear array}
\newacronym{snr}{SNR}{signal-to-noise ratio}
\newacronym{snqr}{SNQR}{signal-to-noise and quantization ratio}
\newacronym{glrt}{GLRT}{generalized likelihood ratio test}
\newacronym{sll}{SLL}{sidelobe level}
\newacronym{socp}{SOCP}{second-order cone program}
\newacronym{si}{SI}{self-interference}
\newacronym{maxsi}{max. SI}{maximum SI}
\newacronym{svd}{SVD}{singular value decomposition}
\newacronym{sdp}{SDP}{semidefinite programming}
\newacronym{scs}{SCS}{subcarrier spacing}
\newacronym{rf}{RF}{radio frequency}
\newacronym{rcs}{RCS}{radar cross-section}
\newacronym{cissir}{CISSIR}{codebooks with integral split for self-interference reduction}

\newacronym{umi}{UMi}{urban microcell}
\newacronym{los}{LoS}{line-of-sight}

\newacronym{fr2}{FR2}{frequency range 2}
\newacronym{dft}{DFT}{discrete Fourier transform}
\newacronym{mmwave}{mmWave}{millimeter wave}

\newacronym{bs}{BS}{base station}
\newacronym{ue}{UE}{user equipment}
\newacronym{csi}{CSI}{channel state information}

\newacronym{ls}{LS}{least-squares}
\newacronym{lmmse}{LMMSE}{linear-minimum-mean-square-error}

\newacronym{papr}{PAPR}{peak-to-average power ratio}
\newacronym{iq}{I/Q}{in-phase/quadrature}
\newacronym{lna}{LNA}{low-noise amplifier}
\newacronym{iid}{i.i.d.}{independent identically distributed}

\newacronym{mimo}{MIMO}{multiple-input multiple-output}

\newacronym{adc}{ADC}{analog-to-digital converter}

\newacronym{ldpc}{LDPC}{low-density parity-check}
\newacronym{ofdm}{OFDM}{orthogonal frequency-division multiplexing}
\newacronym{qpsk}{QPSK}{quadrature phase shift keying}
\newacronym{qam}{QAM}{quadrature amplitude modulation}

\newacronym{kkt}{KKT}{Karush–Kuhn–Tucker}
\newacronym{tdd}{TDD}{time-division duplex}

\newacronym{crlb}{CRLB}{Cram\`{e}r-Rao lower bound}
\newacronym{mumimo}{MU-MIMO}{multiuser MIMO}
\newacronym{zf}{ZF}{zero-forcing}

%
\title{\acs{cissir}: Beam Codebooks with Self-Interference Reduction Guarantees
for Integrated Sensing\\and Communication Beyond 5G}
%
%
%

\author{Rodrigo~Hernang\'{o}mez,~\IEEEmembership{Member,~IEEE,}
        Jochen~Fink,~\IEEEmembership{Member,~IEEE,}
        Renato~Lu\'{i}s~Garrido~Cavalcante,~\IEEEmembership{Member,~IEEE,}
        and~S{\l}awomir~Sta\'{n}czak,~\IEEEmembership{Senior Member,~IEEE}
\thanks{}
}
\IEEEpubid{%
    \begin{minipage}{\columnwidth}
    \vspace{1em}
    \copyright 2025 IEEE. Personal use of this material is permitted. Permission from IEEE must be obtained for all other uses, in any current or future media, including reprinting / republishing this material for advertising or promotional purposes, creating new collective works, for resale or redistribution to servers or lists, or reuse of any copyrighted component of this work in other works.
    \end{minipage}%
    \hspace{1em}%
    \begin{minipage}{\columnwidth}
        $ $
    \end{minipage}
}

%
%

\markboth{Accepted Draft v3.1.3}%
{}
%



\maketitle

\begin{abstract}
We propose a beam codebook design for \gls{isac}
that reduces \gls{si} to alleviate analog distortion.
Our optimization framework, which considers either
tapered beamforming or phased arrays for both analog and hybrid schemes,
modifies given reference codebooks
such that a certain \gls{si} power level is achieved.
In contrast to other low-\gls{si} codebooks,
which often rely on hardly interpretable optimization parameters,
we provide design guidelines to obtain sensing performance guarantees by deriving
analytical bounds on saturation and analog-to-digital quantization
in relation to the multipath \gls{si} level.
By selecting standard reference codebooks in our simulations,
we show how our method substantially improves the \acl{snr} for sensing
with little impact on 5G-\acs{nr} communication.
\end{abstract}


%
\IEEEpeerreviewmaketitle

\glsresetall


\section{Introduction}
%
%
%
%

\label{sec:intro}

\IEEEPARstart{N}{ext-generation} mobile networks are expected to incorporate sensing capabilities as an additional service~\cite{zhang2022enabling}.
To this aim, the paradigm of \gls{isac} is being actively investigated, and
\gls{isac}'s
many challenges and opportunities are rapidly unfolding as a result~\cite{liu2020joint,wymeersch20226g,liu2022integrated}.

\Gls{si} has been identified as one such \gls{isac} challenge, principally,
but not only, in so-called monostatic setups with co-located \gls{tx} and \gls{rx}
sensing systems~\cite{liu2022integrated,girotodeoliveira2022joint,baquerobarneto2019fullduplex,riihonen2023fullduplex}.
\Acl{si} is sometimes addressed by adopting radar-centric \gls{isac} waveforms~\cite{marin2021monostatic}, for which, similarly as for conventional radar, simple analog circuitry can isolate the \gls{rx} echoes from the \gls{tx} signals on the time-frequency plane~\cite{richards2010principles}.
This principle, although proven useful even for
joint design schemes~\cite{xiao2022waveform}, cannot be
generally applied
to the communication-centric waveforms
that mobile networks usually employ, as these are typically dense in time and frequency~\cite{riihonen2023fullduplex,baquerobarneto2019fullduplex}.
As a result, \gls{si} produces an unwanted signal contribution at
communication-centric \gls{isac} receivers, causing adverse effects along the processing chain:
saturation at the \gls{rx} antenna,
desensitization at the \gls{adc}, or
signal masking in the digital baseband~\cite{alexandropoulos2022fullduplex,askar2023lossless}.
Here, the former stages are critical,
as they may introduce nonlinear distortion.
For this reason, a cascaded approach is often adopted,
such that different \gls{si} thresholds are attained at each analog stage
to avoid distortion, and the residual \gls{si} is canceled in the digital domain~\cite{richards2010principles,alexandropoulos2022fullduplex}.

Many strategies to reduce \gls{si} 
for communication waveforms
have been studied in the literature,
often through the lens of the full-duplex paradigm~\cite{baquerobarneto2019fullduplex,riihonen2023fullduplex}.
This includes
passive antenna elements~\cite{askar2023lossless,yalcinkaya2023comparative},
active cancelers~\cite{alexandropoulos2022fullduplex,kolodziej2016multitap},
beamforming~\cite{liu2023fullduplex,hernangomez2024optimized,roberts2023lonestar,bayraktar2023selfinterference,bayraktar2024hybrid,liu2023joint,he2023fullduplex}, or intelligent metasurfaces~\cite{tohidi2023fullduplex};
a comprehensive review can be found in~\cite{smida2024inband}.
While most techniques require
drastic innovations in mobile networks and standards,
\gls{si}-reducing
beamforming may leverage
the existing \gls{mimo} capabilities
in the 5G \gls{nr} standard~\cite{dahlman20185g}
to enable \gls{isac} upon
current infrastructure.
Analog~\cite{liu2023fullduplex,hernangomez2024optimized,bayraktar2023selfinterference,roberts2023lonestar} and hybrid~\cite{bayraktar2024hybrid,liu2023joint} architectures represent
a natural choice here, since digital \gls{rx} beamforming
fails to address all effects happening before the \gls{adc}~\cite{he2023fullduplex}.

Regarding
beamforming and \gls{mimo} processing,
5G \gls{nr} often resorts to
predefined beam codebooks
to reduce channel feedback overhead~\cite{dahlman20185g,fu2023tutorial,dreifuerst2023massive}. 
In this vein, the authors of LoneSTAR~\cite{roberts2023lonestar}
and of~\cite{bayraktar2023selfinterference}
have proposed design methods for codebooks that achieve
considerably low \gls{si}.
These methods, however, rely on hyperparameter tuning and heuristic optimization,
and they are devised with the assumption of
flat-fading \gls{si} channels.
In contrast to this, radio measurements and industrial applications
often deal with multipath \gls{si},
whereby strong clutter reflections are as relevant as
direct antenna coupling despite the longer propagation path~\cite{askar2021interference,he2017spatiotemporal,mandelli2023survey}.

In this paper, we propose a principled
low-\gls{si} beamforming design method
for \gls{isac}, which yields
\gls{cissir}\footnote{Pronounced like ``scissor''.}.
The main contributions can be summarized as follows:
\begin{itemize}
    \item We adopt a multipath \gls{mimo} \gls{si} model
    for \gls{isac} beamforming, which, to the best of our knowledge,
    had not been addressed yet in the literature.
    Our model covers both analog and hybrid beam codebooks,
    while considering realistic \gls{adc} sampling.
    \item We devise a simple optimization framework to design
    codebooks based on either tapered beamforming,
    using matrix decomposition, or phased arrays via \gls{sdp}~\cite{fuchs2014application}.
    In contrast to the state of the art~\cite{bayraktar2023selfinterference,roberts2023lonestar},
    our framework parametrizes the \gls{si} level incurred by the \gls{tx} and \gls{rx} codebooks,
    which we design independently from each other by splitting the \gls{si} channel.
    \item Furthermore, we derive novel analytical bounds that allow us to
    select optimization parameters systematically to meet sensing quality criteria,
    thus endowing \gls{cissir} with feasibility and performance guarantees.
    More specifically, the derived bounds relate the \gls{si} level with the
    \gls{adc} quantization error and the \gls{rx} saturation.
    Our simulations confirm the superior performance of \gls{cissir} over the LoneSTAR codebook~\cite{roberts2023lonestar} according to said criteria,
    while hinting at the compatibility of our method with the  5G-\gls{nr} standard
    via link-level results.
    \item Lastly, we have based our code on popular open-source libraries like
    Sionna\texttrademark~\cite{sionna} and CVXPY~\cite{diamond2016cvxpy}, and we have
    made it publicly available for reproducibility\footnote{ Code available at \href{https://github.com/rodrihgh/cissir}{github.com/rodrihgh/cissir}}.
\end{itemize}

The remainder of this paper is organized as follows:
\cref{sec:sys-opt} introduces the system model and
the optimization problem to design \gls{si}-reduced beam codebooks for \gls{isac}.
\cref{sec:solution} presents the proposed solution, based on
the \gls{tx}-\gls{rx} split of the multipath
\gls{mimo} \gls{si} channel.
\cref{sec:adc} establishes the analytical bounds
of our scheme, providing parameter-selection guidelines
to meet given sensing constraints.
\cref{sec:results} compares \gls{cissir} to
LoneSTAR~\cite{roberts2023lonestar} and
to customary 5G-\gls{nr} codebooks~\cite{dahlman20185g,fu2023tutorial},
via ray-tracing and link-level simulations.
Finally, \cref{sec:conc} delivers some concluding remarks
and a possible outlook of this work.

\subsection{Notation and preliminaries}

We define
$[N]\coloneqq\lbrace1,2,\ldots,N\rbrace$ for a positive integer $N$.
For $z\in\CC$, we indicate its real and imaginary parts as
$\Re(z)$ and $\Im(z)$, and its conjugate as
$z^*=\Re(z)-j\Im(z)$ with $j=\sqrt{-1}$.

We write $\CN(0,\sigma^2)$ to denote the zero-mean complex normal distribution with
variance $\sigma^2$, and $\setU(a,b)$ for the uniform distribution
in $[a,b]$. The expected value of a random variable $X$ is written as $\EE[X]$.

Uppercase bold letters represent matrices
$\V\in\CC^{M\times N}$, with subindexed lowercase bold letters indicating their column vectors, $\V=[\v_1,\ldots,\v_N]$,
($\forall n\in[N]$)
$\v_n\in\CC^{M}$.
We write the convolution of functions $f:\RR\to\CC$ and $g:\RR\to\CC$ as $f\ast g$.
For $p\in[1,\infty]$, the $p$-norm for vectors $\v\in\CC^{M}$ is written as
$\Norm[p]{\v}$, whereas $\norm[p]{f}$ denotes the $p$-norm for functions ($\forall a,b\in\RR$ and $a<b$) $f:[a,b]\to\CC$, to avoid confusion.
Moreover, $\fro{\M}$, $\Norm[2]{\M}$, and $\Norm[\ast]{\M}$ indicate the Frobenius, spectral,
and nuclear norms, respectively, of
$\M\in\CC^{M\times N}$.
For any matrix $\A\in\CC^{M\times N}$, we write its transpose as
$\A^\t$ and its Hermitian transpose as $\A^\h=\left(\A^*\right)^\t$.
For a matrix-valued function $\H:\RR\to\CC^{M\times N}$,
$\c\in\CC^M$, and $\w\in\CC^N$, 
$(\c^\h\H\w)(t)$ stands for the function
$t\mapsto\c^\h\H(t)\w$, and
the integral of $\H(t)$ is understood entry-wise, i.e. 
($\forall n\in[N]\,,\, \forall m\in[M]$):
\begin{align*}
    \left[\int^{t_1}_{t_0}\H(t)\,dt\right]_{m,n}=
    \int^{t_1}_{t_0}\left[\H(t)\right]_{m,n}\,dt\,.
\end{align*}

For any pair of Hermitian matrices
$\A,\B\in\SS^M\coloneqq\lbrace\X\in\CC^{M\times M}|
\,\X=\X^\h\rbrace$, we consider the inner product $\langle\A,\B\rangle\coloneqq\trace(\B^\h\A)={\trace(\B\A)\in\RR}$.
We denote by $\SS^M_+\subset\SS^M$ the $M\times M$
positive semidefinite cone,
which induces the partial order
($\forall \A,\B\in\SS^M$) $\A\succcurlyeq\B\iff\A-\B\in\SS^M_+$.
\section{System model and problem formulation}
\label{sec:sys-opt}
\subsection{System model}
\label{sec:sysmodel}

\begin{figure*}[!ht] 
    \centering
    \def\svgwidth{\textwidth}
    \normalsize
    \includegraphics[width=0.99\textwidth]{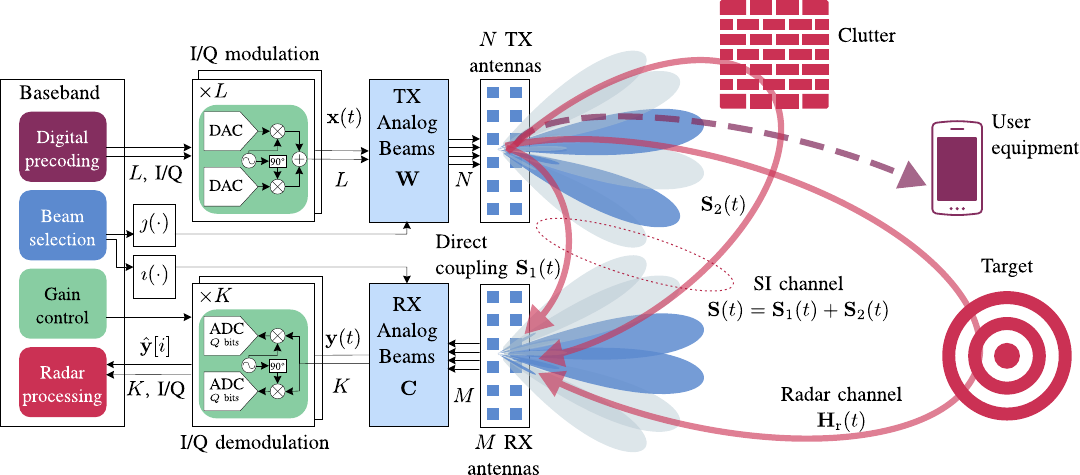}%
    \caption{System model: \gls{isac} beam codebooks in the face of multipath \acf{si}.
    Amplifiers and filters are omitted for simplicity.}%
    \label{fig:sysmodel}%
\end{figure*}

We consider a hybrid-beamforming \gls{mimo} system with $N$ \gls{tx} antennas and $L<N$ \gls{tx} \gls{rf} chains,
as depicted in \cref{fig:sysmodel}.
For a given time slot of length $T$,
the \gls{tx} \gls{rf} chains are fed with $L$ time-continuous signals,
which we write in vector form ($\forall t\in[0,T]$) as
$\x(t)\coloneqq[x_1(t),\,\ldots\,,x_L(t)]$ with
($\forall \ell\in[L]$) $x_\ell:$ $[0,T]\to\CC$.

The signals in $\x(t)$ carry user data and are transmitted according to $L$ selected beamforming vectors out of $L'\geq L$ columns from a given \gls{tx} codebook matrix
$\W\in\CC^{N\times L'}$.
The \gls{tx} chain accepts
any suitable existing mechanism
for beam selection,
denoted as the mapping $\jmath:\;[L]\to[L']$.
That is, ($\forall \ell\in[L]$)
the beamformer $\w_{\jmath(\ell)}$ is applied to $x_\ell(t)$.

Moreover, the system is equipped with an \gls{rx} sniffer~\cite{wild20236g} dedicated to radar sensing,
whereas \gls{tdd} communication is assumed
in accordance with
5G \gls{fr2}~\cite{dahlman20185g}.
The \gls{rx} sniffer
has $M$ antennas and $K$ \gls{rf} chains, such that
($\forall k\in[K]$) 
the $k$-th \gls{rx} signal is beamformed
by $\c_{\imath(k)}$, a vector selected from
the $K'\geq K$ columns in the \gls{rx} codebook $\C\in\CC^{M\times K'}$.
The \gls{rx} beam mapping 
$\imath:\;[K]\to[K']$ is envisioned 
to match the \gls{tx}
beam directions given by $\jmath$
to maximize the \gls{rx} echo power~\cite{liu2023joint}.

The \gls{tx} and \gls{rx} codebooks
are restricted to the feasibility sets
$\Wl$ and $\Ck$, i.e.:
\begin{align*}
    \W\in\Wl&\coloneqq\left\lbrace\W\in\CC^{N\times L'}\middle|
(\forall\ell\in[L'])\;\w_\ell\in\setW\right\rbrace\,,\\
\C\in\Ck&\coloneqq\left\lbrace\C\in\CC^{M\times K'}\middle|
(\forall k\in[K'])\;\c_k\in\setC\right\rbrace\,,
\end{align*}
where $\setW\in\lbrace\Wd,\Wp\rbrace$
and $\setC\in\lbrace\Cd,\Cp\rbrace$ express
constraints for either tapered beamforming
or phased arrays, the two most common array technologies.
In the tapered case,
the codebook columns are only subject to unit power:
\begin{equation}
\label{eq:1power}
\Wd\coloneqq\left\lbrace\w\in\CC^{N}\middle|
\Norm[2]{\w}=1\right\rbrace,\,
\Cd\coloneqq\left\lbrace\c\in\CC^{M}\middle|
\Norm[2]{\c}=1\right\rbrace.
\end{equation}
On the other hand, for phased arrays we have
\begin{align*}
    \Wp&\coloneqq\left\lbrace \w\in\CC^{N}\,\middle|
    \;(\forall n\in[N])\;\norm{[\w]_{n}}={1}/{\sqrt{N}}\right\rbrace,\\
    \Cp&\coloneqq\left\lbrace\c\in\CC^{M}\,\middle|
    \;(\forall m\in[M])\;\norm{[\c]_{m}}={1}/{\sqrt{M}}\right\rbrace.
\end{align*}
We note that $\Wp\subset\Wd$ and $\Cp\subset\Cd$, i.e., beam unit power is
assumed for both beamforming variants.

The signals in
$\x(t)$ impinge on the \gls{rx} sniffer over a \gls{mimo} channel
$\H:[0,\Tc]\to\CC^{M\times N}$ with maximum delay $\Tc$, which we write
($\forall t\in[0,\Tc],\forall m\in [M], \forall n\in[N]$)
as:
\begin{equation}
    \H(t)\coloneqq\begin{bmatrix}
    h_{1,1}(t) & \dots  & h_{1,N}(t)\\
    \vdots & \ddots & \vdots\\
    h_{M,1}(t) & \dots  & h_{M,N}(t) 
    \end{bmatrix},\;h_{m,n}:[0,\Tc]\to\CC\,.\label{eq:mimo-ch}
\end{equation}
Furthermore, we assume $\H(t)$ to be the superimposition of
the unknown radar channel
$\Hr:[0,\Tc]\to\CC^{M\times N}$ and
the \gls{si} channel
$\S:[0,\Tc]\to\CC^{M\times N}$:
\begin{equation*}
    (\forall t\in[0,\Tc])\;
    \H(t)=\Hr(t)+\S(t)\,.
\end{equation*}
We assume $\S(t)$ is known through estimation procedures like those
described in~\cite{roberts2023lonestar}, whereby the \gls{si}
is separated from the radar channel
in the delay-Doppler domain by exploiting the short-range
and static nature of $\S(t)$~\cite{askar2021interference,roberts2022beamformed}.
While our analysis holds for arbitrary bounded matrix-valued functions,
we choose to model $\S(t)=\sum_{d=1}^D\S_{d}(t)$ as a sum of $D$ paths.
This choice, motivated by the
available measurements in the literature~\cite{he2017spatiotemporal},
allows us to differentiate between direct antenna coupling and clutter reflections
($\S_1(t)$ and $\S_2(t)$ in \cref{fig:sysmodel})
for an easier comparison with single-path approaches.
All entries in $\Hr(t)$ and $\S(t)$ are
functions with domain in $[0,\Tc]$,
similarly to the entries of $\H(t)$ in \eqref{eq:mimo-ch}.

The \gls{rx} signals are expressed in vector form as $\y(t)$:
\begin{align}
\nonumber
(\forall t\in[0,T+\Tc])\quad
    \y(t)&=\left[y_1(t),\ldots,y_K(t)\right]^\top,\;\mathrm{where}\\
    (\forall k\in[K])\quad
    y_k(t)=\sum_{\ell=1}^L&\left(\left(\c_{\imath(k)}^\h\H\,\w_{\jmath(\ell)}\right)\ast x_\ell\right)(t)\,.
    \label{eq:rx-signal}
\end{align}
For all $k\in[K]$, two \glspl{adc} sample the \gls{iq} components of $y_k(t)$ at the rate $1/T_\mathrm{S}$ with $Q$ quantization bits each, yielding
the following \gls{rx} digital sensing signal
($\forall i\in\lbrace0,1,\ldots,\left\lfloor{(T+\Tc)}/{T_{\mathrm{S}}}\right\rfloor\rbrace$):
\begin{align}
    \nonumber
    \yH[i]&=[\yh_1[i],\ldots,\yh_K[i]]^\top,\\
    \left(\forall k\in [K]\right)\quad
    \yh_{k}[i]&\coloneqq y_k(i\Ts)+v_{k}[i]
    +u_{k}[i].\label{eq:digital}
\end{align}
In \eqref{eq:digital},
$v_{k}[i]\sim\CN(0,\sigma^2_\mathrm{T})$ is a sample of \gls{iid} thermal noise, whereas
$u_{k}[i]$ is an independent sample of complex quantization noise.
The latter is described in more detail in \cref{sec:adc}.

\subsection{The optimization problem}

Given a pair of
\gls{tx}/\gls{rx} beam codebooks
$\A\in\Wl$ and $\B\in\Ck$,
which we assume optimal for
 \gls{tdd} codebook-based communication,
we aim to derive
a codebook pair
($\W,\C$) that guarantees
a certain sensing quality
in the presence of \gls{si}.
We assess sensing quality via the \gls{snr},
which directly determines
detection and estimation performance~\cite{richards2010principles}.
In practice,
\gls{si} and \gls{snr} are linked
through the quantization noise
$\sgQ^2\coloneqq\EE[\norm{u_k[i]}^2]$,
as we show in \cref{sec:adc}.
In particular, \cref{thm:quant-bound} therein
allows
to bound $\sgQ$, up to a scaling factor,
by the \emph{\gls{maxsi}} $\mS(\C, \W)$, which
we define for $\S(t)$
($\forall\,\C\in\Ck$, $\forall\,\W\in\Wl$)
as
\begin{equation}\label{eq:maxsi}
    \mS(\C,\W)\coloneqq\max_{\substack{k\in[K'],\\\ell\in[L']}}\int_0^{\Tc}\norm{\c_k^\h\S(t)\w_{\ell}}\,dt\,.
\end{equation}
Therefore, we choose to limit the
\gls{si}-induced
quantization noise
by imposing a certain \emph{target \gls{si}}
${\varepsilon}>0$ such that $\mS(\C, \W)\leq\varepsilon$.
We note that  \eqref{eq:maxsi} does not depend on the
\gls{tx} signal
$\x(t)$ or the beam-selection mappings $\imath$ and $\jmath$, since these may be unknown during codebook design.

Moreover, we wish to optimize communication
by minimizing $\stx^2$, the deviation of
the \gls{isac} \gls{tx} codebook $\W$ from the
communication-optimal
codebook $\A$.
We also aim to minimize the deviation $\srx^2$ between $\B$
and the \gls{rx} codebook $\C$,
which could be potentially used
for standard \gls{tdd} communication.
Additionally, designing ($\W,\C$) close to ($\A,\B$)
may improve the sensing \gls{snr},
since practical communication codebooks usually maximize the
\gls{tx}/\gls{rx} gain across a grid of angular directions~\cite{dahlman20185g,fu2023tutorial}.
In other words,
 $\stx^2$ and  $\srx^2$ minimization
tends to increase the
numerator of
$\mathrm{SNR}\propto\norm{\c^\h_{k_\theta}\b_\theta\a^\h_\theta\w_{\ell_\theta}}^2$
for a target at the direction given by the steering vectors
$\a_\theta\in\CC^N$,
$\b_\theta\in\CC^M$,
and for certain unknown beam indices $\ell_\theta\in[L']$,
$k_\theta\in[K']$.

In line with the literature~\cite{roberts2023lonestar,bayraktar2023selfinterference}, we
compute $\stx^2$ and $\srx^2$ with the Frobenius norm, i.e.:
\begin{equation}
\begin{aligned}
    \stx^2&\coloneqq\fro{\W-\A}^2/\underbrace{\fro{\A}^2}_{=L'}=\frac{1}{L'}
    \sum_{\ell=1}^{L'}
    \Norm[2]{\w_\ell-\a_\ell}^2\,,\\
    \srx^2&\coloneqq\fro{\C-\B}^2/\underbrace{\fro{\B}^2}_{=K'}=\frac{1}{K'}
    \sum_{k=1}^{K'}
    \Norm[2]{\c_k-\b_k}^2\,.
    \label{eq:cb-dev}
\end{aligned}
\end{equation}
Hence, the design of ($\W,\C$)
can be posed as the following optimization problem:
\begin{subequations}\label[problem]{prob:joint}
\begin{align}\label{eq:joint-obj}
    \underset{{\W\in\Wl,\,\C\in\Ck}}{\minimize}
    &\fro{\W-\A}^2+
    \omega\fro{\C-\B}^2\\
    \st\quad&\label{eq:eps-const}
     \mS(\C,\W)\leq{\varepsilon}\,,
\end{align}
\end{subequations}
with $\omega>0$ a trade-off parameter
between \gls{tx} and \gls{rx} codebook deviation.
\cref{prob:joint} inverts the approach taken 
in~\cite{roberts2023lonestar,bayraktar2023selfinterference},
where the (flat-fading) \gls{si}
is minimized under a heuristically tuned maximum codebook deviation
constraint
$\scb^2\geq\max\lbrace\stx^2,\srx^2\rbrace$.
In contrast to this,
\cref{prob:joint} enables both
the tuning of $\varepsilon$ according to sensing metrics
(cf. \cref{sec:adc}), and the mitigation of
multipath \gls{si}, via $\mS(\C,\W)$.
\section{Proposed scheme}
\label{sec:solution}

Problem~\eqref{prob:joint} is nonconvex
because of the constraint sets $\Ck$ and $\Wl$, and also due to the coupling of the codebooks $\C$ and $\W$ in $\mS(\C,\W)$.
In order to derive a more tractable problem, we propose a decoupling of $\C$ and $\W$.
The resulting problem, albeit still nonconvex,
can be solved optimally for the case of tapered beamforming and phased arrays,
as we detail in \crefrange{sec:taper}{sec:phased}.

As a preliminary step,
we note that we can modify \eqref{eq:joint-obj},
given the unit power constraint \eqref{eq:1power},
to write \cref{prob:joint} in the following
equivalent form:
\begin{align}\label[problem]{prob:vec}
    \underset{{\W\in\Wl,\,\C\in\Ck}}{\maximize}
    &\frac{1}{L'}
    \sum_{\ell=1}^{L'}
    \Re\left(\a_\ell^\h\w_\ell\right)+
    \frac{\omega}{K'}
    \sum_{k=1}^{K'}
    \Re\left(\b_k^\h\c_k\right)\\
    \st\;\;
    (\forall k\in[K'],\,&\forall\ell\in[L'])\quad
     \int_0^{\Tc}\norm{\c_k^\h\S(t)\w_{\ell}}\,dt
     \leq{\varepsilon}\,.\nonumber
\end{align}

\subsection{Problem decoupling}
\label{sec:split}

We propose an alternative to \eqref{prob:vec} that allows us to optimize the columns of
$\C$ and $\W$ individually while ($\C$, $\W$) remains in the original feasible region. For this,
we consider the following decomposition, which we refer to as \emph{integral split}.

\begin{definition}[Integral split]
\label{def:svd}
We define the integral split of
an \gls{si} channel $\S(t)$
as the matrices
$\Gt\in\SS^N_+$ and
$\Gr\in\SS^M_+$:
\begin{equation*}
    \Gt\coloneqq\int_0^{\Tc}\V_t\Eig_t\V_t^\h\,dt\,,
    \;\mathrm{and}\;
    \Gr\coloneqq\int_0^{\Tc}\U_t\Eig_t\U_t^\h\,dt\,,
\end{equation*}
where, for every $t\in[0, \Tc]$,
$\U_t$, $\Eig_t$, and $\V^\h_t$
are given as a \gls{svd} of $\S(t)=\U_t\Eig_t\V^\h_t$.
\end{definition}
\begin{proposition}\label{prop:paralel}
    Let $\Gt\in\SS^N_+$ and $\Gr\in\SS^M_+$ be
    the integral split (see \cref{def:svd})
    of $\S(t)$ in \eqref{eq:maxsi}, and
    consider the quantity
    ($\forall\,\C\in\Ck,\,\forall\,\W\in\Wl$):
    \begin{equation*}
        \mgg(\C,\W)\coloneqq
        \max_{k\in[K']}\sqrt{\c_k^\h\Gr\c_k}
        \max_{\ell\in[L']}\sqrt{\w_\ell^\h\Gt\w_\ell}\,.
    \end{equation*}
    \begin{equation}
        \text{Then, we have:}\quad
        \label{eq:split-ineq}
        \mS(\C,\W)\leq
        \mgg(\C,\W)\,.\qquad\qquad
    \end{equation}
\end{proposition}
\begin{proof}
    The proof is given in \cref{apx:split}. 
\end{proof}

Using \cref{prop:paralel},
we suggest replacing $\mS(\C,\W)$ with
$\mgg(\C,\W)$ in the equivalent original \cref{prob:vec}
to yield:
\begin{align}\label[problem]{prob:relax}
    \underset{{\W\in\Wl,\,\C\in\Ck}}{\maximize}
    &\frac{1}{L'}
    \sum_{\ell=1}^{L'}
    \Re\left(\a_\ell^\h\w_\ell\right)+
    \frac{\omega}{K'}
    \sum_{k=1}^{K'}
    \Re\left(\b_k^\h\c_k\right)\\
    \st\;\;
    (\forall k\in[K'],\,&\forall\ell\in[L'])\;\;
     \sqrt{\c_k^\h\Gr\c_k}
     \sqrt{\w_\ell^\h\Gt\w_\ell}
     \leq{\varepsilon}\,.\nonumber
\end{align}
\cref{prob:relax} is a restriction of
\cref{prob:joint}, in the sense that
any feasible codebook pair for \eqref{prob:relax}
is also feasible for 
\eqref{prob:joint}
by \cref{prop:paralel}.
Due to the potential bound gap in \eqref{eq:split-ineq},
however, a solution to \eqref{prob:relax} is
not necessarily a solution to \eqref{prob:joint}.

Nevertheless, \eqref{eq:split-ineq}
can become tight for some practical situations.
For the case $M=N$,
for instance, tightness is approximated
if the \gls{maxsi} is reached for a certain
pair $\c=\w$
under an \gls{si} channel containing a predominant
symmetric path $\S(t_0)=\s\s^\h$, such that
$\Gt\approx\Gr\approx\s\s^\h$.
Thus, we expect
$\mS(\C,\W)\approx\mgg(\C,\W)$
under spatially sparse symmetric \gls{si} channels,
which are relevant for monostatic \gls{isac}~\cite{baquerobarneto2019fullduplex,hernangomez2024optimized}.

Based on \cref{prob:relax},
we propose designing
a codebook pair
($\bar{\C},\bar{\W}$) columnwise,
by solving $L'+K'$
problems in the following form:
\begin{subequations}\label[problem]{prob:split}
\begin{align}
    \label{eq:split.obj}
    \underset{\z\in\CC^{P}}{\maximize}&{\quad\Re(\r^\h\z)}\\
    \st\quad&\quad\label{eq:si-split}
    \z^\h\G\z\leq\epsilon,\\
    &\quad
    \z\in\setZ\,,
\end{align}
\end{subequations}
with $P\in\lbrace M,N\rbrace$,
$\setZ\in\lbrace\setW,\setC\rbrace$,
and the parameters
$\r\in\setZ$,
$\G\in\SS^P_+$, and $\epsilon>0$
chosen
($\forall \ell\in[L']$,
$\forall k\in[K']$)
as
\begin{align}
    \nonumber
    \r=\a_\ell\,,\, \G=\Gt\,,\,\;
\epsilon=\varepsilon\beta,
    \;\mathrm{for}\;
    \bar{\w}_\ell&,\\
    \r=\b_k\,,\, \G=\Gr\,,\,
\epsilon=\varepsilon\beta^{-1},\,
    \mathrm{for}\;\bar{\c}_k&,
\label{eq:opt-params}
\end{align}
for $\varepsilon>0$ in \eqref{prob:relax} and
$\beta>0$, such that
\begin{equation*}
    \mgg(\bar{\C},\bar{\W})=
    \Bigl(\max_{k\in[K']}\underbrace{\bar{\c}_k^\h\Gr\bar{\c}_k}_{\leq\varepsilon\beta^{-1}}
    \max_{\ell\in[L']}
    \underbrace{\bar{\w}_\ell^\h\Gt\bar{\w}_\ell}_{\leq\varepsilon\beta}\Bigr)^{1/2}
    \leq\varepsilon\,.
\end{equation*}
Denoting the optimal value of \eqref{prob:split}
with $f_{\setZ}(\r,\G,\epsilon)$, we can
express the codebook deviations attained through this scheme as
${\stx^2=1-\frac{1}{L'}\sum_{\ell=1}^{L'}f_{\setW}(\a_\ell,\Gt,\varepsilon\beta)}$ and
${\srx^2=1-\frac{1}{K'}\sum_{k=1}^{K'}f_{\setC}(\b_k,\Gr,\varepsilon\beta^{-1})}$.
Optimization theory can prove that $f_{\setZ}(\r,\G,\epsilon)$ is a monotonically increasing
function of $\epsilon$~\cite{boyd2004convex}, and thus $\stx^2$ and $\srx^2$
are monotonically decreasing on $\varepsilon$ for $\beta>0$.

The possibility to
trade off \gls{tx} and \gls{rx}
codebook deviation,
embodied by $\omega$
in \eqref{prob:joint},
is provided instead
by $\beta$
in \eqrefrange{prob:split}{eq:opt-params}.
In symmetric scenarios where $\A=\B$
and $\Gt=\Gr$,
we expect the balance point $\beta=1$
to yield a good sensing \gls{snr}
for given $\varepsilon$, based
on the observation that
the \gls{rx} echo power
for a target at
an angular direction $\theta$
is proportional to the product
$f_\setC(\b_{k_\theta},\Gr,\varepsilon\beta^{-1})f_\setW(\a_{\ell_\theta},\Gt,\varepsilon\beta)$
for certain $k_\theta=\ell_\theta\in[L']$.
The value of $\beta$ can be alternatively determined with respect to antenna saturation,
while $\varepsilon$ can be set in relation
to the quantization noise.
The choice of $\varepsilon$ and $\beta$ with respect to
said criteria is analyzed in detail in \Cref{sec:adc}.

We refer to the solutions to
\eqrefrange{prob:split}{eq:opt-params}
as \emph{\acrfull{cissir}}.
In the following \crefrange{sec:taper}{sec:phased}, we particularize
Problem~\eqref{prob:split} for specific
choices of $\setC$ and $\setW$, but
we first remark some general properties
induced by the unit power constraint
\eqref{eq:1power} from the system model in
\cref{sec:sysmodel}.

\begin{remark}\label{rm:cvx}
\cref{prob:split} is still nonconvex
for $\setZ\in\lbrace\Wd,\Wp,\Cd,\Cp\rbrace$
due to \eqref{eq:1power}.
A straightforward
convex relaxation can be obtained by replacing
$\setZ$ with its convex hull.
\end{remark}
\begin{remark}\label{rm:max-val}
The optimal value fulfills
$f_{\setZ}(\r,\G,\epsilon)\leq1$,
and the equality is attained with
$\zsol=\r$ if $\epsilon\geq\r^\h\G\r$.
To see this, we note that:
\begin{equation*}
    (\forall\zsol\in\setZ)\quad 
    \Re(\r^\h\zsol)\leq
    \norm{\r^\h\zsol}\leq
    \Norm[2]{\r}\Norm[2]{\zsol}=1\,,
\end{equation*}
where the equality follows from \eqref{eq:1power}.
The inequalities above become tight for $\zsol=\r$, but
$\r$ is only feasible if $\epsilon\geq\r^\h\G\r$.
\end{remark}

\begin{algorithm}[t]%
\caption{Tapered \acs{cissir} optimization.}%
\label{alg:lagrange}%
\begin{algorithmic}[1]%
\Require $\S(t),\A,\B,\varepsilon,\beta$
\State $\Gt,\Gr\gets$ \Call{IntegralSplit}{$\S(t)$}
\Comment{Def.~\ref{def:svd}}
\State  $\W\gets$ \Call{SpectralCodebook}{$\A,\Gt,\varepsilon\beta$}
\State  $\C\gets$ \Call{SpectralCodebook}{$\B,\Gr,\varepsilon\beta^{-1}$}
\State  \Return $\W$, $\C$
\Function{SpectralCodebook}{$\R,\G,\epsilon$}
    \State $\Q,\Gam \gets$ \Call{SpectralDecomposition}{$\G$}
    \Comment{Prop.~\ref{pr:taper-sol}}
    \For{$i=1,2,\,\ldots\,,\,J\coloneqq$\Call{NumColumns}{$\G$}}
    \State $\r\gets[\R]_{:,i}$
    \If {$\epsilon<\r^\h\G\r$}
        \State  $\z_i\gets$ \Call{OptimizeVector}{$\r,\epsilon,\Q,\Gam$}
    \Else
        \State $\z_i\gets\r$
    \EndIf
\EndFor
\State \Return $[\z_1,\z_2,\,\ldots\,,\z_{J}]$
\EndFunction
\Function{OptimizeVector}{$\r,\epsilon,\Q,\Gam$}
\State $\setP(\nu)\gets$ \Call{BuildPoly}{$\r,\epsilon,\Q,\Gam$}
\Comment{Prop.~\ref{pr:taper-sol}}
\State $\nu_\star\gets$ \Call{LargestRealRoot}{$\setP(\nu)$}
\State $\check{\z}\gets\Q(\Gam+\nu_\star\I)^{-1}\Q^\h\r$
\State \Return ${\check{\z}}/{\Norm[2]{\check{\z}}}$
\EndFunction
\end{algorithmic}%
\end{algorithm}

\subsection{Tapered beamforming}
\label{sec:taper}

We first particularize \cref{prob:split} to tapered beamforming,
where amplitude and phase can be controlled continuously and individually for every antenna element.
As we will see, this case admits a semi-closed
form solution that relies only on matrix decompositions and other linear algebra operations.

For tapering, we have $\setW=\Wd$ and $\setC=\Cd$, and \cref{prob:split} takes the form:
\begin{subequations}\label[problem]{prob:taper}
\begin{align}
    \label{eq:obj-taper}
    \underset{\z\in\CC^{P}}{\maximize}&{\quad\Re(\r^\h\z)}\\
    \st\quad&\quad\label{eq:si-taper}
    \z^\h\G\z\leq\epsilon,\\
    &\quad\z^\h\z=1\,.\label{eq:norm-const}
\end{align}
\end{subequations}
Although \cref{prob:taper} is nonconvex,
it can be solved efficiently
as stated below and outlined in
\cref{alg:lagrange}.
\begin{proposition}\label{pr:taper-sol}
    \renewcommand{\theenumi}{\roman{enumi}}%
    Consider a spectral decomposition $\G=\Q\Gam\Q^\h\neq\Zero$ of the matrix $\G\in\SS^P_+$ in \eqref{eq:si-taper}, where $\Q$ is unitary and $\Gam$ a diagonal matrix with diagonal entries given by $(\forall p\in[P])\;
        \sigma_p\coloneqq
        \left[\Gam\right]_{p,p}$.
    \cref{prob:taper} admits
    a solution
    \begin{equation*}
        \zsol=\check{\z}/\Norm[2]{\check{\z}}\,,\quad\mathrm{with}\quad
        \check{\z}\coloneqq\Q
        \left(\Gam+\nu_\star\I\right)^{-1}
        \Q^\h\r\,,
    \end{equation*}
    where $\nu_\star$ is a positive root of the real polynomial
        \begin{equation*}
    \setP(\nu)\coloneqq\sum_{p=1}^P{\left(\sigma_p-\epsilon\right)\norm{\q_p^\h\r}^2}\prod_{q\neq p}{\left(\sigma_q+\nu\right)^2}\,,
\end{equation*}
under certain sufficient conditions depending on the rank of $\G$.
In particular, $\zsol$ exists if either:
\begin{enumerate}
    \item\label{cond:rank-def}
    $\rank(\G)<P$ and $\epsilon\in(0,\r^\h\G\r)$, or
    \item\label{cond:rank-full}
    $\rank(\G)=P$ and $\epsilon\in(\bar\sigma(\r,\G),\r^\h\G\r)$, with
 \begin{equation*}\label{eq:emin}
    \bar\sigma(\r,\G)\coloneqq
    \frac{\sum_{p=1}^P{{\norm{\q_p^\h\r}^2}/{\sigma_p}}}
    {\sum_{p=1}^P{{\norm{\q_p^\h\r}^2}/{\sigma_p^2}}}
    \,.
 \end{equation*}
\end{enumerate}
\end{proposition}
\begin{proof}
The proof is given in \cref{apx:taper}.
\end{proof}

\begin{algorithm}[t]%
\caption{Phased \acs{cissir} optimization.}
\label{alg:sdr}
\begin{algorithmic}[1]
\Require $\S(t),\A,\B,\varepsilon,\beta$
\State $\Gt,\Gr\gets$ \Call{IntegralSplit}{$\S(t)$}
\Comment{Def.~\ref{def:svd}}
\State  $\W\gets$ \Call{SdpCodebook}{$\A,\Gt,\varepsilon\beta$}
\State  $\C\gets$ \Call{SdpCodebook}{$\B,\Gr,\varepsilon\beta^{-1}$}
\State  \Return $\W$, $\C$
\Function{SdpCodebook}{$\R,\G,\epsilon$}
    \For{$i=1,2,\,\ldots\,,\,J\coloneqq$\Call{NumColumns}{$\G$}}
    \State $\r\gets[\R]_{:,i}$
    \State${\tilde{\Z}}\gets$ \Call{SolveSdp}{$\r,\G,\epsilon$}\label{alg:solve-sdr}
    \Comment{\eqrefrange{eq:sdp-obj}{eq:psd}}
    \State $\tilde{\z}\gets\sqrt{\sigma_{\max}({\tilde{\Z}})}\v_{{\max}}({\tilde{\Z}})$
    \State $\z_i\gets\tilde{\z}({\tilde{\z}^\h\r}/{\norm{\tilde{\z}^\h\r}})$
\EndFor
\State \Return $[\z_1,\z_2,\,\ldots\,,\z_{J}]$
\EndFunction
\end{algorithmic}
\end{algorithm}

\subsection{Phased-array optimization}
\label{sec:phased}

In contrast to the previous section,
the phased-array case does not admit a direct solution.
Nevertheless, we will show that phased-array codebooks
can be obtained via \acf{sdp}~\cite{fuchs2014application},
as described in \cref{alg:sdr}.

Phased array codebooks $\Cp$ and $\Wp$ yield the
following form of \cref{prob:split}:
\begin{equation}\label[problem]{prob:phased}
\begin{aligned}
    \underset{\z\in\CC^{P}}{\maximize}&{\quad\Re(\r^\h\z)}\\
    \st\quad&\quad
    \z^\h\G\z\leq\epsilon,\\
    (\forall p\in[P])&\quad
    \z^\h\e_p\e_p^\t\z
    ={1}/{P}\,,
\end{aligned}
\end{equation}
where $\e_p=\left[\I\right]_{:,p}$ is the $p$-th standard basis vector.
We will now propose a method to solve \eqref{prob:phased}
via \gls{sdp}.

\begin{proposition}
A solution to \cref{prob:phased}
is given by
\begin{equation*}
    \z_\star=\tilde{\z}\frac{\tilde{\z}^\h\r}{\norm{\tilde{\z}^\h\r}},\quad
    \tilde{\z}=
    \sqrt{\sigma_{\max}({\tilde{\Z}})}\v_{{\max}}({\tilde{\Z}})\,,
\end{equation*}
where
${\sigma_{\max}({\tilde{\Z}})}$ and
$\v_{{\max}}({\tilde{\Z}})$
are the largest eigenvalue
and the corresponding eigenvector, respectively,
of a solution
$\tilde{\Z}$ to the following nonconvex \gls{sdp} problem:
\begin{subequations}\label[problem]{prob:sdp}
\begin{align}
    \label{eq:sdp-obj}
    \underset{\Z\in\SS^P}{\maximize}&{\quad
    \langle \r\r^\h, \Z \rangle}\\
    \st\quad&\quad
    \langle \G, \Z \rangle\leq\epsilon,\\
    (\forall p\in[P])&\quad\label{eq:sdp-mag1}
   \langle \e_p\e_p^\h, \Z \rangle={1}/{P}\,,\\
   \label{eq:psd}
   &\quad \Z\succcurlyeq\Zero\,,\\
   &\quad\rank(\Z)\leq1\,.\label{eq:rank1}
\end{align}
\end{subequations}
\end{proposition}
\begin{proof}
We first note that
any solution $\zsol$ to \eqref{prob:phased} must fulfill
$\Re(\r^\h\zsol)=\norm{\r^\h\zsol}$.
This follows from the fact that, for any feasible $\z_\mathrm{f}$,
we can construct another feasible $\tilde{\z}_\mathrm{f}$ such that
\begin{equation*}
    \Re(\r^\h\z_\mathrm{f})\leq
    \norm{\r^\h\z_\mathrm{f}}=
    \Re(\r^\h\tilde{\z}_\mathrm{f})\,,\quad
    \tilde{\z}_\mathrm{f}=\z_\mathrm{f}\left({\z_\mathrm{f}^\h\r}/{\norm{\z_\mathrm{f}^\h\r}}\right).
\end{equation*}
Consequently, we can substitute the objective function in \eqref{prob:phased}
with $\norm{\r^\h\z}^2=\z^\h\r\r^\h\z$.
In addition to this, we note the trace circular property
($\forall\M\in\SS^M$, $\forall\z\in\CC^{M}$):
\begin{equation*}
    \z^\h\M\z=
    \trace(\z^\h\M\z)=
    \trace({\z\z^\h}\M)=
    \langle\M,\z\z^\h\rangle\,,
\end{equation*}
which allows us to express \cref{prob:phased}
in terms of inner products as \eqref{prob:sdp},
such that $\tilde{\Z}=\zsol\zsol^\h$
and $\zsol$ solves \eqref{prob:phased}.
\end{proof}

A convex relaxation of \cref{prob:sdp} can be obtained
by dropping the nonconvex constraint \eqref{eq:rank1}.
A solution $\tilde{\Z}$ to the convex \gls{sdp}
problem
\eqrefrange{eq:sdp-obj}{eq:psd} is not guaranteed to have rank one,
but
we can further elaborate on $\rank(\tilde{\Z})$ by relaxing
constraint~\eqref{eq:sdp-mag1}:
\begin{equation}\label[problem]{prob:sdpack}
\begin{aligned}
    \underset{\Z\in\SS^P}{\maximize}&{\quad
    \langle \r\r^\h, \Z \rangle}\\
    \st\quad&\quad
    \langle \G, \Z \rangle\leq\epsilon,\\
    (\forall p\in[P])&\quad
   \langle \e_p\e_p^\h, \Z \rangle\leq{1}/{P}\,,\\
   &\quad \Z\succcurlyeq\Zero\,.
\end{aligned}
\end{equation}
\cref{prob:sdpack} is also known as
a \emph{semidefinite packaging} problem,
and Sagnol~\cite{sagnol2011class} proved that it
always has
a rank-one solution.
Therefore, the solution $\tilde{\Z}$
to \eqrefrange{eq:sdp-obj}{eq:psd} will also have rank one
if it coincides with a solution to
\eqref{prob:sdpack} for
given $\r$, $\G$, and $\epsilon$.

Furthermore,
we know that
$\exists\,\epsilon_\star\leq\r^\h\G\r$ such that
$\forall\epsilon\geq\epsilon_\star$
$\rank(\tilde{\Z})=1$ (cf. \cref{rm:max-val}). Unfortunately, $\epsilon_\star$ cannot be easily
narrowed down further, but we can still assess the rank of $\tilde{\Z}$
through the ratio between its spectral and nuclear norms:
\begin{equation}\label{eq:rank1metric}
    \Upupsilon(\tilde{\Z})\coloneqq\frac{\|\tilde{\Z}\|_2}{\|\tilde{\Z}\|_\ast}=
    \frac{\sigma_{\max}({\tilde{\Z}})}{\trace(\tilde{\Z})}\in\left[{1}/{P}\,,\;1\right],
\end{equation}
for which we know that $\Upupsilon(\tilde{\Z})=1\iff\rank(\tilde{\Z})=1$.
In \cref{sec:results}, we use this fact
to evaluate the feasibility of \cref{alg:sdr} in our scenarios.

\section{Performance bounds and parameter selection}
\label{sec:adc}

The proposed \crefrange{alg:lagrange}{alg:sdr} in \cref{sec:solution} aim to
reduce the \gls{maxsi}
$\mS(\C,\W)$
in \eqref{eq:maxsi} according to the
parameters $\varepsilon$ and $\beta$.
We analyze now the effect of \gls{si}
before the \glspl{adc} for the hybrid-beamforming system in \cref{sec:sysmodel},
in order to provide recommendations for the choice of said parameters.

In \cref{sec:bound}, we derive a theoretical
bound on the variance of the quantization noise introduced by the \gls{adc},
while we also take into account
other dynamic-range limitations
and saturation.
In \cref{sec:snr}, we extend the derived bound to the
\gls{snr} for analog \gls{ofdm} radar,
a popular \gls{isac} setup~\cite{zhang2022enabling,liu2020joint,wymeersch20226g,liu2022integrated,baquerobarneto2019fullduplex,riihonen2023fullduplex,girotodeoliveira2022joint}.
Finally, in \cref{sec:si-selection} we discuss how to select the target
\gls{si} $\varepsilon$ based on the derived bounds 
and our signal knowledge.

\subsection{Self-interference impact at the
\texorpdfstring{\acl{adc}}{ADC}}
\label{sec:bound}

\glspl{adc} have limited dynamic range,
which can be
critical
for \gls{isac} in the face of \gls{si}.
In particular, the \gls{adc} bit resolution $Q$ causes
signal quantization, modeled
as $u_{k}[i]$ in \eqref{eq:digital}
by following the well-established assumption
of uniform white quantization noise~\cite{gray1990quantization}.
That is,  we assume $u_{k}[i]$ is
independent of the \gls{adc}'s input signal,
with \gls{iid} $\Re(u_{k}[i])$ and $\Im(u_{k}[i])$ 
drawn from $\setU(-\frac{\Delta}{2},\frac{\Delta}{2})$. 
The quantization step $\Delta$ can
then be expressed in terms of
the \gls{adc}'s full-scale value $\yfs$ and the bit resolution $Q$ as $\Delta={2\yfs}/{2^Q}$ , for which the
quantization noise power equals
\begin{equation}\label{eq:qnoise}
\sgQ^2\coloneqq\EE[\norm{u_{k}[i]}^2]=2\frac{\Delta^2}{12}=
\frac{2\yfs^2}{3}{2^{-2Q}}=
\bQ\yfs^2.
\end{equation}

We write $\bQ$ in \eqref{eq:qnoise}
to stress the relation
$\sgQ^2\propto\yfs^2$, as well as to
extend our analysis to other sources
of dynamic-range limitations,
such as logarithmic \glspl{adc}
or finite-precision arithmetic.
With this in mind,
we can minimize $\sgQ^2$
via gain control
by setting $\yfs$ to the smallest possible value above saturation.
Under our \gls{tdd} assumption,
we can leverage
the absence of uplink signals and
the full knowledge of both the transmitted
$\x(t)$ and the calibrated \gls{si}
channel $\S(t)$ to select
\begin{equation}
\label{eq:yfs}
    \yfs=\gamma\max_{
    \substack{t\in[0,T+\Tc]\\ k\in[K]}}
    \norm{\sum_{\ell=1}^L\left(\left(\c_{\imath(k)}^\h\S\,\w_{\jmath(\ell)}\right)\ast x_\ell\right)(t)}\,,
\end{equation}
where we choose a back-off term
$\gamma>1$ to account for the unknown radar channel
$\Hr(t)$.
The condition to avoid saturation is
\begin{equation*}
    \gamma\geq
    \frac{\max_{t,k}
    \norm{\sum_{\ell=1}^L\left(\left(\c_{\imath(k)}^\h\left(\Hr+\S\right)\w_{\jmath(\ell)}\right)\ast x_\ell\right)(t)}}
    {\max_{t,k}
    \norm{\sum_{\ell=1}^L\left(\left(\c_{\imath(k)}^\h\S\,\w_{\jmath(\ell)}\right)\ast x_\ell\right)(t)}}\,.
\end{equation*}
Therefore, we can set $\gamma$ fairly close to 1 in practice,
since the contribution from the \gls{si} channel
$\S(t)$
is typically orders of magnitude larger than the
backscattered power over $\Hr(t)$.

We will now show how
the sampling scheme in
\eqrefrange{eq:qnoise}{eq:yfs}
yields an upper bound on
the \gls{adc} noise as a function of the \gls{maxsi}. The derived bound
only depends on generic
\gls{tx}-signal
properties,
namely the
\gls{tx} power $\Pt$ and
the \gls{papr} of each signal,
($\forall\ell\in[L]$)
$\rho_{x_\ell}$:
\begin{align}
    \label{eq:Pt}
    \Pt&\coloneqq\sum_{\ell=1}^L\frac{1}{T}\int_0^T\norm{x_\ell(t)}^2\,dt=
    \frac{1}{T}\sum_{\ell=1}^L
    \norm[2]{x_\ell}^2\,,\\
    \label{eq:papr}
    \rho_{x_\ell}&\coloneqq T{\infnorm{x_\ell}^2}/{\norm[2]{x_\ell}^2}\,.
\end{align}

\begin{lemma}\label{thm:quant-bound}
Consider the system model in \cref{sec:sys-opt},
with the \gls{maxsi} $\mS(\C,\W)$ in \eqref{eq:maxsi},
the \gls{rx} signal $\y(t)$ in \eqref{eq:rx-signal},
sampled at $1/T_{\mathrm S}$ to obtain
$\yH[i]$ in \eqref{eq:digital}.
Assuming quantization noise,
\gls{rx} gain control as per \eqrefrange{eq:qnoise}{eq:yfs},
and the \gls{tx}-signal properties in \eqrefrange{eq:Pt}{eq:papr},
then we have:
\begin{align}
    \nonumber
    &\EE\left[\Norm[2]{\yH[i]-\y(iT_\mathrm{S})}^2\right]=
    \sum_{k=1}^K\left(\sgT^2+
    \sgQ^2\right)\leq\\
    &K\sgT^2+
    K\Pt{\gamma^2}\bQ
    \mS^2(\C,\W)
   \sum_{\ell=1}^L\rho_{x_\ell}.\label{eq:noise-bound}
\end{align}
\end{lemma}
\begin{proof}
The proof is given in \cref{apx:bound}.
\end{proof}
\begin{corollary}\label{thm:discrete}
Let us assume the same system model as for \cref{thm:quant-bound}, but with the channel propagation represented
by discrete convolution with $\bar{\H}[i]=\bar{\Hr}[i]+\bar{\S}[i]$, i.e., ($\forall k\in[K]$):
\begin{equation*}
    \bar{y}_k[i]=\sum_{\ell=1}^L\left(\left(\c_k\bar{\H}\w_\ell\right)\ast\bar{x}_\ell\right)[i],
    \;\;
    (\forall \ell\in[L])\;
    \bar{x}_\ell[i]
    \coloneqq x_\ell(iT_\mathrm{S})\,.
\end{equation*}
Then, \eqref{eq:noise-bound} still holds with $J=\left\lfloor{(T+\Tc)}/{T_{\mathrm{S}}}\right\rfloor$ and
$\mathfrak{m}_{\bar{\S}}(\C,\W)=\max_{k,\ell}\sum_{i=1}^{J}
\norm{\c_k\bar{\S}[i]\w_\ell}$.
\end{corollary}
\begin{proof}
    The proof follows from the fact that Young's inequality in
    \cref{thm:quant-bound}'s proof
    also applies to the discrete
    convolution~\cite[Theorem (20.18)]{hewitt2012abstract}.
\end{proof}

\cref{thm:quant-bound} allows us to assess
the quantization noise induced
by multipath \gls{si}
for given beamforming codebooks,
which we validate in
\cref{sec:results}
via simulations,
thanks to \cref{thm:discrete}.
These results also provide us with a rule to select the
target \gls{si} $\varepsilon$ in Problems \eqref{prob:joint} and \eqref{prob:split}, as we discuss
in \cref{sec:si-selection}.

As for $\beta$ in \eqref{eq:opt-params},
while $\beta=1$ is
a sensible choice (cf. \cref{sec:split}),
other values of $\beta$ may be preferred
if the system is susceptible to \gls{rx} antenna saturation.
In this sense, we provide the following \cref{pr:sat} as
a rule to choose $\beta$ for given $\varepsilon$
so that ($\forall m\in[M]$) the peak power at the
$m$-th \gls{rx} antenna remains
below a certain $\Psat$, i.e.:
\begin{equation}\label{eq:sat}
    \infnorm{\mathring{y}_m}^2\leq\Psat\,,
    \;\;\mathring{y}_m(t)\coloneqq
    \sum_{\ell=1}^L\left(\left(\e_{m}^\t\H\,\w_{\jmath(\ell)}\right)\ast x_\ell\right)(t)\,.
\end{equation}

\begin{proposition}\label{pr:sat}
Assume the system model from \cref{thm:quant-bound}
and the optimization scheme given by
\eqrefrange{prob:split}{eq:opt-params}.
Then, the condition \eqref{eq:sat} is satisfied if
\begin{equation*}
    \varepsilon\beta=\frac{\Psat}{\gamma^2\Pt
   \sum_{\ell\in[L']}\rho_{x_\ell}
   \max_{m\in[M]}{[\Gr]_{m,m}}}\,.
\end{equation*}
\end{proposition}
\begin{proof}
    Noting that
    $\lbrace\mathring{y}_m(t)\rbrace_{m=1}^M$ is a special case of \eqref{eq:rx-signal}
    for $K=K'=M$ and $\C=\I\in\CC^{M\times M}$,
    the proof follows from the same arguments
    as in \cref{thm:quant-bound}'s proof,
    and from
    the application of \cref{prop:paralel} to
    $\mS(\I,\W)$.
\end{proof}

\subsection{Signal-to-noise ratio for analog
\texorpdfstring{\acs{ofdm}}{OFDM}
radar}
\label{sec:snr}

\cref{thm:quant-bound} provides a bound on noise variance as a function of
the \gls{maxsi}
$\mS(\C,\W)$. Based on this bound, we focus now on
analog beamforming and \gls{ofdm} radar
to extend our performance analysis
to the sensing \gls{snr}, which is the main performance indicator for
sensing detection and estimation~\cite{richards2010principles}.

Under analog beamforming, we have $K=L=1$ \gls{rf} chains.
As a result, we only transmit a signal $\x(t)=x(t)\in\CC$, and
we receive $\y(t)=y(t)\in\CC$ as:
\begin{equation}
    \label{eq:analog}
    y(t)=
    ((\underbrace{\co^\h\S\wo}_{\eqqcolon\hsi}+\underbrace{\co^\h\Hr\wo}_{\eqqcolon\hr})\ast x)(t)=
    \ysi(t)+\yr(t)\,,
\end{equation}
where $\co\in\CC^M$ and $\wo\in\CC^N$ form our selected 
\gls{tx}-\gls{rx} beamforming pair.
This beamforming pair yields
the single-antenna equivalent channels $\hsi(t)$ and $\hr(t)$, which we
assume passband filtered according to the
signal bandwidth $B$.

Furthermore, if we employ \gls{ofdm} waveforms, we can assume our \gls{tx} signal to be stationary with almost flat
power spectral density $X(f)$~\cite{baquerobarneto2019fullduplex,richards2010principles}. In other words, ($\forall f\in[\fc-B/2,\fc+B/2]$)
$X(f)\approx\Pt/B$, and hence:
\begin{align}
    \EE[\norm{\yr(t)}^2]=\zeta\frac{\Pt}{B}\norm[2]{\hr}^2&\approx
    \int_{\fc-B/2}^{\fc+B/2}\frac{\Pt}{B}|H_{\mathrm r}(f)|^2\,df\,,\;
    \label{eq:ypow}
\end{align}
with $\zeta\approx1$. Assuming favorable Nyquist conditions
so that $\EE[\norm{\yr(i\Ts)\,}^2]=\EE[\norm{\yr(t)}^2]$, we can
combine \eqref{eq:ypow} with \cref{thm:quant-bound} 
to bound
the \gls{snr} of
$\yr(t)$ from \eqref{eq:digital} as:
\begin{equation}
    \label{eq:snr}
    \mathrm{SNR}=
    \frac{\EE[\norm{\yr(t)}^2]}{\sgT^2+\sgQ^2}\geq
    \frac{{\zeta}\norm[2]{\hr}^2\Pt/B}
    {\sgT^2+\gamma^2\bQ\mS^2(\C,\W)\Pt\rho_x}\,,
\end{equation}
which relates to the \gls{crlb} for range estimation with
a frequency-flat matched filter~\cite{richards2010principles} as
\begin{equation}
    \label{eq:crlb}
    \mathrm{CRLB}=
        \frac{c^2}{4\,\mathrm{SNR}\cdot B^2_{\mathrm{rms}}}=
    \frac{c^2}{4\,\mathrm{SNR}\cdot B^2\pi^2/3}\,,
\end{equation}
with $c$ the speed of light.
The expression above considers Gaussian noise,
which is not true in general for $u_k[i]$. However,
we can invoke the central limit theorem
if the matched filter is long enough,
as it is expected for \gls{ofdm} radar~\cite{baquerobarneto2019fullduplex,girotodeoliveira2022joint},
to assume that the filtered noise is
approximately Gaussian. Hence,
we may still use \eqref{eq:crlb} in practice to assess sensing
performance
via \eqref{eq:snr} with respect to the \gls{maxsi}
$\mS(\C,\W)$
and the target \gls{si} $\varepsilon$.

Nevertheless, choosing $\varepsilon$
from \eqref{eq:snr} is challenging in practice, as
we often lack the following information:
\begin{enumerate}
    \item The radar channel $\Hr(t)$, and thus $\hr(t)$.
    Even if we know $\Hr(t)$, we cannot compute $\norm[2]{\hr}$
    before optimizing ($\co$, $\wo$).
    However,
    we can approximate it from the reference beamforming pair
    ($\bo$, $\ao$) as 
    $\norm[2]{\hr}\approx\norm[2]{\bo^\h\Hr\ao}$,
    if the codebook deviations $\stx^2$ and $\srx^2$ in \eqref{eq:cb-dev} are low enough.
    \item The \gls{papr} $\rho_x$, since the
    signal $x(t)$ may also be unknown during optimization.
\end{enumerate}
To address this issue,
we suggest a more practical approach
in \cref{sec:si-selection}.

\subsection{Selection of the target self-interference}
\label{sec:si-selection}

To set the target \gls{si}
$\varepsilon$ in \crefrange{alg:lagrange}{alg:sdr},
we first suggest fixing a target quantization noise
$\sigma_\star^2$, either from \eqref{eq:snr} if
$\norm[2]{\hr}^2$ is known, or with respect to
the thermal noise $\sgT^2$ otherwise.
From \cref{thm:quant-bound},
we select the target \gls{si} as
\begin{equation}
    \label{eq:si-choice}
    {\varepsilon}
    =
    \frac{\sigma_\star}
    {\gamma\sqrt{\bQ\Pt\brx
    }}\,,
\end{equation}
which guarantees $\sgQ^2\leq\sigma_\star^2$.
The estimated \gls{papr} $\brx>0$
is computed from the available information on $\x(t)$ during optimization.
Some options to compute $\brx$ include:

\subsubsection{Matched optimization}
In the best-case scenario where we do know $\x(t)$,
we can just choose $\brx=\sum_{\ell=1}^L\rho_{x_\ell}$ in \eqref{eq:si-choice} and solve \eqref{prob:split} only for the vectors given
by $\imath$ and $\jmath$.
This option may be relevant if the optimization algorithm is fast enough to be run for every new \gls{tx} symbol, e.g., if it has
a semi-closed form as in \cref{sec:taper}.

\subsubsection{Average optimization}
\label{sec:avg-opt}
In certain setups, it might suffice to ensure that the quantization noise variance
is below $\sigma_\star^2$ only for the average \gls{tx} symbol. In this case,
the codebooks can be optimized a priori with
$\brx=\sum_{\ell=1}^L\EE[\rho_{x_\ell}]$.

\subsubsection{Quantile optimization}
If a certain noise level must be guaranteed, yet memory and speed limitations rule out the previous methods,
then we can make $\brx=\sum_{\ell=1}^L\EE[\rho_{x_\ell}]/\alpha$.
This \gls{papr} estimate controls the probability of
a \gls{tx} symbol to exceed the noise threshold via Markov's inequality, i.e.:
\begin{equation*}
    \Pr(\sgQ^2\geq\sigma^2_\star)\leq\alpha\,.
\end{equation*}

\section{Numerical results}
\label{sec:results}

\begin{table}[!t]
\caption{Simulation parameters}%
\begin{tabular}{lrlc}
\textbf{NAME}& \multicolumn{2}{c}{\textbf{VALUE}}          & \textbf{SYMBOL}          \\
\multicolumn{4}{c}{A. \textbf{Signal parameters}}                                                     \\
Carrier frequency         & 28                   & \si{\giga\hertz}                  & $\fc$           \\
Wavelength                & 10.7                 & \si{\milli\meter}                  & $\uplambda$   \\
Waveform       & (cyclic-prefix)                 & \acs{ofdm}                &                       \\
5G NR numerology          & 3                    & \cite{dahlman20185g}                     & $\upmu$\\
Subcarrier spacing        & 120                  & \si{\kilo\hertz}                 & \acs{scs}                      \\
Symbol duration           & 9.44                 & \si{\micro\second}                   & $T$                      \\
Max. channel delay          & 1.11                 & \si{\micro\second}                   & $T_\mathrm{c}$\\
\multicolumn{4}{c}{\textbf{B. Antenna parameters}}                                                    \\
 \acs{tx} antenna elements               & 8                    &                      & $N$                      \\
\acs{rx} antenna elements               & 8                    &                      & $M$\\
\acs{tx} codebook size          & 27                    &                      & $L'$                     \\
\acs{rx} codebook size          & 27                    &                      & $K'$\\
\acs{tx}  \acs{rf} chains              &  1 or 4                    &                      & $L$\\
\acs{rx}  \acs{rf} chains              & 1 or 4                   &                      & $K$\\
Ant. element spacing           & 5.35                 & \si{\milli\meter}                   & $\uplambda/2$ \\
Ant. element pattern           &  TR 38.901                & \cite{tr38901}                  &  \\
Elevation angle                & 0    & \si{\degree}                     &                          \\
Sector coverage                & 120    & \si{\degree}                     &                          \\
\multicolumn{4}{c}{\textbf{C. \acs{si} parameters}}                                                    \\
Min.  \acs{tx}- \acs{rx} separation     & 5.35                 & \si{\milli\meter}                   & $\uplambda/2$ \\
Wall distance             & 4                    & \si{\meter}                    &                          \\
Wall azimuth              & 65                   & \si{\degree}                    &                          \\
Multipath \acs{maxsi}             & -54.5                & \si{\deci\bel}                      & $\mS(\B,\A)$\\
Single-path  \acs{maxsi}      & -60.6                & \si{\deci\bel}                   & $\mSo(\B,\A)$\\
\multicolumn{4}{c}{\textbf{D. Sensing parameters}}                                                    \\
 \acs{tx} power         & 30                   & \si{\deci\belmilliwatt}                  & $\Pt$                    \\
Sampling interval         & 0.51                 & \si{\nano\second}                   & $T_\mathrm{S}$\\
Number of  symbols         & 1000                   &                      &                          \\
Number of subcarriers     & 1680                 &                      &                          \\
Bandwidth                 & 200                  & \si{\mega\hertz}                 &      $B$                    \\
Modulation                & 64                   & \acs{qam}                  &                          \\
Target azimuth            & -39                  & \si{\degree}                      &                          \\
Target distance           & 40                   & \si{\meter}                    &                          \\
Target radar cross section & 1                    & \si{\meter\squared} &                          \\
Target  \acs{rx} power           & -81.0                & \si{\deci\belmilliwatt}                    &                          \\
Thermal noise             & -90.8                & \si{\deci\belmilliwatt}                    & $\sigma_{\mathrm{T}}^2 $ \\
 \acs{adc} bits& 6                    & bit                  & $Q$\\
\multicolumn{4}{c}{\textbf{E. Communication   parameters}}                                            \\
\acs{ue} antennas               & 1                    &                      &                          \\
Number of subcarriers     & 420                  &                      &                          \\
Number of symbols         & 14                   &   (2 pilots)    &                          \\
Bandwidth                 & 50                   & \si{\mega\hertz}                   &                          \\
Modulation                & 4                    & \acs{qam}                  &                          \\
Code rate                  & 0.7                  &                      &                          \\
Bitrate                   & 62.3                & \si{\mega\bit\per\second}              &                          \\
Channel model                 & \acs{los} outdoor \acs{umi}                    & \cite{tr38901}                     &                        \end{tabular}%
\label{tab:sysparams}%
\end{table}

We corroborate now the presented theoretical analysis
of \gls{cissir} with a series of simulation experiments.
First, we present some sensing results in \cref{sec:sens},
which illustrate and validate the performance bounds
derived in \cref{sec:adc}. Afterwards,
we study the impact of \gls{si} reduction on
communication, in \cref{sec:comm}.
We continue with an
investigation of the communication-sensing trade-off
in \cref{sec:tradeoff}, both for \gls{cissir}
and for the baseline LoneSTAR~\cite{roberts2023lonestar}.
Finally, we complete our numerical evaluation with
a discussion on the time complexity and feasibility
of the studied algorithms.

Our simulation environment has been developed from previous work in~\cite{dehghani2024analog} and is based on Sionna\texttrademark, a Python library by Nvidia that integrates link-level and
ray-tracing simulation tools~\cite{sionna}.
Our reference standard is 5G \gls{nr} \gls{fr2}.
As such, we employ
\gls{ofdm} waveforms in the \gls{mmwave} frequency with \gls{qam}, and \gls{scs} according to a given numerology $\upmu$~\cite{dahlman20185g}. Likewise,
our reference beam codebooks
$\A$ and $\B$
are given by the columns
of a \gls{dft} matrix with
an oversampling factor
$O_1=4$~\cite{ts38214}. We assume \glspl{ula} with $M=N=8$ antennas,
radiation pattern according
to TR 38.901~\cite{tr38901},
fixed elevation, and $L'=K'=27$ beams,
corresponding to a full grid
over a sector coverage of \SI{120}{\degree}.
\cref{tab:sysparams} contains the complete parameter list.

We simulate the \gls{si} channel via ray-tracing, with $D=2$ paths corresponding
to $\S_1(t)$ and $\S_2(t)$
(see \cref{fig:sysmodel}). The component $\S_1(t)$ is produced by a spherical wave model,
similar to the flat \gls{mimo} \gls{si} channel in \cite{roberts2023lonestar,bayraktar2023selfinterference,bayraktar2024hybrid}, while we add the clutter
$\S_2(t)$ as a wall-reflected path. Leveraging \cref{thm:discrete}, we approximate
the continuous channels by their discrete versions sampled at $T_\mathrm{S}\ll 1/B$, and we apply a passband filter to account for the sensing
bandwidth $B$.

We optimize the codebooks via \crefrange{alg:lagrange}{alg:sdr}
for tapered beamforming and phased arrays, respectively.
We set $\beta=1$ (and thus $\epsilon=\varepsilon$), which yields the best \gls{rx} echo power
in our setup, as discussed in \cref{sec:split}.
Unless otherwise noted, we use phased arrays as obtained by \cref{alg:sdr}.
The convex \gls{sdp} problem \eqrefrange{eq:sdp-obj}{eq:psd}
therein is solved with CVXPY and its splitting conic solver~\cite{diamond2016cvxpy}, which we also use
to optimize LoneSTAR.

\subsection{Sensing performance}
\label{sec:sens}

\begin{figure}[!t]
     \centering
     \input{figs/range_profile.pgf}%
    \caption{\gls{ofdm} range profiles with reference and
    \gls{si}-optimized
    \gls{tx}/\gls{rx} analog
    beamformers pointing towards
    a target at \SI{-39}{\degree}.}%
    \label{fig:range-profile}%
\end{figure}

Our primary goal, as detailed in
\crefrange{sec:sys-opt}{sec:solution}, is to
design
\gls{isac} beam codebooks
for which
the \gls{si}-induced
quantization noise
at our \gls{rx} sniffer
is reduced. To this aim,
we have derived the
performance bound \eqref{eq:snr}
for the sensing \gls{snr}
in \cref{sec:adc}, which
we compare now
with experimental results.

For this, we simulate
a backscatter channel $\H(t)=\S(t)+\Hr(t)$
in a single-target scenario,
and we apply the channel to an
\gls{ofdm} signal $x(t)$ (cf. \cref{tab:sysparams}.D)
to obtain both $y(t)$, as per \eqref{eq:analog},
and its noisy version $\yh(t)$,
according to the quantization scheme in \cref{sec:bound}.
From $\yh(t)$,
we compute a range profile via
matched filtering with $x(t)$,
as we display in \cref{fig:range-profile}
for two distinct beam pairs pointing towards the target.
More specifically,
one range profile is obtained
for a reference beam pair
($\bo$, $\ao$) under a \gls{maxsi} around \SI{-55}{\deci\bel},
such that the target reflection lies
just below the system's dynamic range.
In contrast to this, we estimate a second
range profile with optimized beams
($\co$, $\wo$) for a target \gls{si}
$\varepsilon$ near \SI{-85}{\deci\bel},
where we observe that the reduced noise floor reveals
the target.

This observation
encourages us to compare
the theoretical \gls{snr}
bound \eqref{eq:snr} with
the simulated \gls{snr}
from $\yh(t)$ and the associated square-root \gls{crlb},
for different values of $\varepsilon$.
The comparison is shown
in \cref{fig:snr},
where we replace the \gls{maxsi} $\mS(\C,\W)$
in \eqref{eq:snr} with $\varepsilon$,
we set $\norm[2]{\hr}=\norm[2]{\bo\Hr\ao}$,
we compute the \gls{papr} $\rho_s\approx$~\SI{9.6}{\deci\bel}
as the average over symbols (cf. \cref{sec:avg-opt}),
and we make $\gamma=1$ and $\zeta\approx0.98$.
We remind the reader that this bound
is an a priori approximation of \eqref{eq:snr}
that is only valid
for the average noise over
\gls{tx} symbols, and
for high \gls{si} such that
$\norm[2]{\hr}=\norm[2]{\co^\h\Hr\wo}\approx\norm[2]{\bo^\h\Hr\ao}$.

In \cref{fig:snr}, the \gls{si}
constraint \eqref{eq:si-split} only
becomes active for ($\co,\wo$) when the target \gls{si}
$\varepsilon$ is below \SI{-68}{\deci\bel}.
In that regime, we observe a gap of at most
\SI{2.5}{\deci\bel}
between the sensing \gls{snr}, when the noise is symbol-averaged, 
and the bound in \eqref{eq:snr}. This bound gap, which arises from Young's inequality
(cf. \cref{apx:bound}), compensates to some extent
the \gls{snr} fluctuation caused by the \gls{papr}
variance over symbols.
As expected, this theoretical bound becomes invalid for very low
\gls{si}, due to a gain loss induced by high codebook deviations
$\stx^2=\fro{\W-\A}^2/\fro{\A}^2$ and $\srx^2=\fro{\C-\B}^2/\fro{\B}^2$.

\begin{figure}[!t]
     \centering
     \input{figs/bound_sqnr.pgf}%
    \caption{Simulated sensing
    \gls{snr} and \gls{crlb} (symbol-averaged noise and \SI{99}{\percent} interval), with approximate bound for high \gls{si}, in function of the target \gls{si}
    $\varepsilon$.}%
    \label{fig:snr}%
\end{figure}

\begin{figure}[!b]
     \centering
    \input{figs/beams.pgf}%
    \vspace{-2mm}%
    \caption{Orthogonal beams from the 
    \gls{tx} reference $\A$ ($\varepsilon=$\SI{-54.5}{\deci\bel})
    and from \gls{cissir} $\W$ ($\varepsilon=$\SI{-100.5}{\deci\bel}),
    with color-matched nominal beam direction.}%
    \label{fig:beams}%
\end{figure}

To illustrate codebook deviation,
we compare
some reference \gls{tx} orthogonal beams from $\A$ with those from $\W$ in \cref{fig:beams}.
Here, we have chosen very low $\varepsilon$,
which leads to
higher secondary lobes and gain loss in the main lobe,
in the best case.
In the worst case,
beam misalignment occurs,
as can be seen near the sector edges at \SI{\pm45}{\degree}.
As a result, the angular coverage of \gls{cissir}
may become as low as
\SI{90}{\degree}, in contrast to
the reference \SI{120}{\degree} coverage,
as we decrease
the \gls{maxsi}.

Motivated by this phenomenon, we investigate next the effect of
\gls{si} reduction and
codebook deviation on communication.

\subsection{Communication performance}
\label{sec:comm}

We evaluate communication performance
via link-level simulations from
a \gls{bs} towards single-antenna \glspl{ue},
placed randomly on an \gls{umi} sector under
\gls{los} outdoor conditions~\cite{tr38901} and 2 beamforming scenarios.
In particular, we either assume a single
user served with analog beamforming ($L=1$),
or hybrid \gls{mumimo} for 2 users and
$L=4$ beams.
The \gls{bs} selects the $L$ strongest
orthogonal beams from the \gls{tx} codebook,
according to \cite[Algorithm 1]{fu2023tutorial},
and feeds them with the
\ac{zf} precoded data streams.
In turn, the \glspl{ue} perform
\gls{ls} channel estimation from 2 \gls{ofdm} pilot symbols, and
\gls{lmmse} equalization before decoding.
The data in the remaining 12 \gls{tx} symbols is
\gls{ldpc}-encoded
and constitutes a single block
(see \cref{tab:sysparams}.E for more details).

\begin{figure}[!t]
    \centering
    \input{figs/bler.pgf}%
    \caption{\acs{bler} for different analog and hybrid \gls{cissir} 
    codebooks. The solid lines correspond to the \gls{tx} reference
    codebook.}%
    \label{fig:bler}%
\end{figure}

We simulate the \gls{bler} as a function
of the energy per bit to noise power spectral density ($E_b/N_0$)
for the
5G-\gls{nr} \gls{tx} reference codebook, as well as for three \gls{cissir}
\gls{tx}
codebooks with different levels of
target \gls{si} $\varepsilon$ and
codebook deviation $\stx^2$.
The resulting \gls{bler} curves in~\cref{fig:bler}
show a progressive degradation of the communication performance
as we reduce $\varepsilon$,
and conversely increase $\stx^2$ (cf. \cref{sec:split}).
In the analog scenario, this degradation seems to have
a small impact on communication for
small values of $\stx^2$, which allows us to add
\SI{16}{\deci\bel} of extra \gls{si} attenuation
while losing at most \SI{1}{\deci\bel}
of \gls{tx} power for communication.
For the optimized range profile in \cref{fig:range-profile},
on the other hand, the analog \gls{bler} curve in \cref{fig:bler}
lies within \SI{3}{\deci\bel} apart from the reference codebook.
As we further decrease the target \gls{si}
$\varepsilon$, however, this distance from the reference
becomes as large as \SI{6}{\deci\bel}, as a result of the greater codebook deviation.

We can also observe
communication degradation
for the \gls{mumimo} results
in \cref{fig:bler} as
the \gls{tx} deviation increases.
This deterioration is, nevertheless,
negligible for $\stx^2$ below
\SI{-11}{\deci\bel}, and the performance
gap for $\stx^2\approx$\SI{-8}{\deci\bel}
is smaller than in the analog case.
Hybrid \gls{mumimo} appears thus to be less
sensitive to the chosen beam codebook,
which might be caused by the
inherent constraints of hybrid beamforming,
such as rank-reduced precoding~\cite{fu2023tutorial}.

From these experiments, we conclude that
extreme \gls{si} reduction may degrade
conventional communication,
as $\stx^2$ increases and gives rise to
beam misalignment and poorer angular coverage.
That being said, communication degradation
seems to unfold gently as the target \gls{si} is reduced,
hinting at a benign operating point
where sensing-enhanced, yet communication-transparent codebooks
can be realized.

\begin{figure*}[!t]
     \centering
     \subfloat[Single-path \gls{si} channel $\S_1(t)$.]{
         \centering
         \input{figs/1tap_si.pgf}%
         \label{fig:1tap-opt}%
     }%
     \subfloat[Multipath \gls{si} channel $\S(t)$.]{
         \centering
         \input{figs/full_si.pgf}%
         \label{fig:full-opt}%
     }%
    \caption{\gls{tx} and \gls{rx} codebook deviations
    $\stx^2$ and $\srx^2$
    as functions of the
    \gls{maxsi} for LoneSTAR and \gls{cissir}, both for phased and tapered arrays.}%
    \label{fig:tradeoff}%
\end{figure*}

\subsection{Trade-off between self-interference
and codebook deviation}
\label{sec:tradeoff}

From the analysis in
\crefrange{sec:sens}{sec:comm},
we consider \gls{maxsi} and codebook deviations
good indicators for sensing and
communication performance, respectively.
Thus, we use them now to study the
sensing-communication
trade-off of \gls{cissir} and compare it
to that of our baseline LoneSTAR~\cite{roberts2023lonestar}.

LoneSTAR codebook optimization, described in
\cite[Algorithm 1]{roberts2023lonestar}, 
assumes a frequency-flat \gls{mimo}
\gls{si}  channel that can be modeled
as a single matrix $\bar{\H}\in\CC^{M\times N}$.
Given $\bar{\H}$, it minimizes the following \gls{si} quantity:
\begin{equation}\label{eq:lonestar}
    \fro{\C\bar{\H}\W}^2
    =\sum_{k=1}^{K'}\sum_{\ell=1}^{L'}
    \norm{\c_k\bar{\H}\w_\ell}^2\,,
\end{equation}
under per-antenna power constraints
and a maximum
codebook deviation $\scb^2\geq\max\lbrace\stx^2,\srx^2\rbrace$.
The authors follow an alternating approach
so that the \gls{tx}
codebook $\W$ is optimized first
and then plugged in \eqref{eq:lonestar}
for the optimization of the \gls{rx} codebook
$\C$.
In addition to this,
the codebooks are projected to their
feasibility sets, $\setW\ni\W$ and
$\setC\ni\C$, after the corresponding
optimization step.

For a fair comparison,
we first consider a simplified flat-fading
\gls{si} model that only includes the direct coupling $\S_1(t)$,
and we set $\bar{\H}=\S_1(\tau_1)$ in \eqref{eq:lonestar}
for LoneSTAR optimization, with $\tau_1$ the mean path delay.
For \gls{cissir}, we simply set $\S(t)=\S_1(\tau_1)\delta(t-\tau_1)$
in \crefrange{alg:lagrange}{alg:sdr}.

We optimize LoneSTAR and \gls{cissir} for different values
of $\scb^2$ and target \gls{si} $\varepsilon$,
respectively, and plot the resulting
trade-off curves in
\cref{fig:1tap-opt},
assuming either 
phased
or tapered arrays.
Here, we observe that \gls{cissir} can reduce \gls{si}
further than LoneSTAR for the same $\stx^2$ and $\srx^2$,
with a \gls{maxsi} gap going up to \SI{14}{\deci\bel}
for a codebook deviation around \SI{-18}{dB}.
We suppose this performance gap lies on
LoneSTAR's beam-exhaustive
\gls{si} sum-power minimization, via \eqref{eq:lonestar},
which is sensible for full-duplex communication
but overly conservative for \gls{isac}, as per \cref{thm:quant-bound}.
An additional trade-off analysis using \eqref{eq:lonestar}
instead of $\mSo(\C,\W)$ seems to confirm this supposition,
as it displays an \gls{si} sum-power up to \SI{5}{\deci\bel} lower for LoneSTAR
than for \gls{cissir}.

In \cref{fig:1tap-opt},
\gls{cissir} reaches around \SI{2}{\deci\bel} more
\gls{si}-reduction with tapered ($\Wd, \Cd$)
than with phased arrays ($\Wp, \Cp$) for the same deviation,
which is expected since $\Wp\subset\Wd$ and $\Cp\subset\Cd$.
For LoneSTAR, tapering allows for more \gls{si} reduction,
but it performs similarly
to phased arrays for \gls{maxsi} over \SI{-85}{\deci\bel}, perhaps
due to the considered per-antenna power constraint~\cite{roberts2023lonestar}.
\gls{cissir} curves coincide
for \gls{tx} and \gls{rx},
which is expected for $\beta=1$.
In fact, a different choice of $\beta>0$ would
simply alter the parameter
$\epsilon\in\lbrace\varepsilon\beta,\varepsilon\beta^{-1}\rbrace$
in \cref{prob:split},
causing a mere horizontal curve shift along the \gls{maxsi} axis.

\begin{figure*}[!b]
     \centering
     \subfloat[LoneSTAR, $\mSo(\C,\W)=$\SI{-69.9}{\deci\bel}.]{
         \centering
         \includegraphics{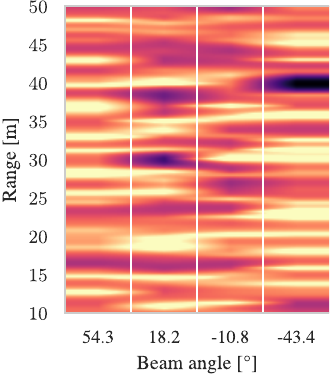}%
         \label{fig:ra-lst}%
     }\hfill
     \subfloat[\acs{cissir}, $\mSo(\C,\W)=$\SI{-80.6}{\deci\bel}.]{
         \centering
         \includegraphics{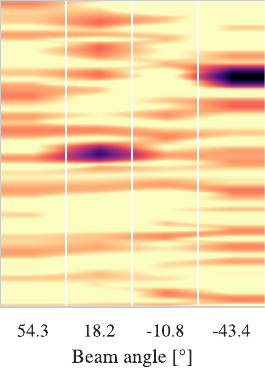}%
         \label{fig:ra-csr}%
     }\hfill
     \subfloat[Ground truth ($\A,\B$ without \gls{si} or quantization).]{
         \centering
         \includegraphics{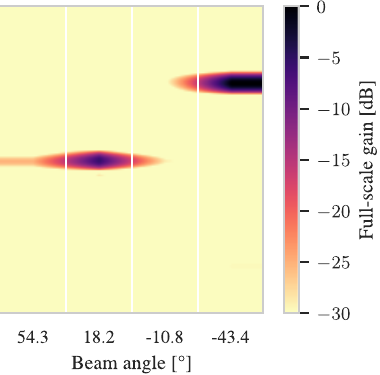}%
         \label{fig:ground_truth}%
     }%
    \caption{Range-angle maps for $L=K=4$ beams with different codebooks and digital residual \gls{si} cancellation. \gls{tx} deviation $\stx^2\approx$\SI{-18.2}{\deci\bel} in
    \crefrange{fig:ra-lst}{fig:ra-csr}.}%
    \label{fig:rng-ang}%
\end{figure*}

We plot the trade-off curves again
for our full multipath \gls{si} model $\S(t)$
in~\cref{fig:full-opt}.
Since LoneSTAR only considers
flat-frequency \gls{si},
we retain the
$\S_1(t)$-optimized
LoneSTAR codebooks above to
evaluate their performance under
multipath \gls{si}.
Unsurprisingly,
LoneSTAR experiences a performance degradation in this multipath
scenario, as it is unable to
attain a \gls{maxsi}
below \SI{-83}{\deci\bel}.
On the other hand,
\gls{cissir} can still
reduce the \gls{maxsi} well below \SI{-100}{\deci\bel}
and \SI{-120}{\deci\bel} for phased and tapered
arrays, respectively.
While the authors in \cite{roberts2023lonestar}
do not directly address multipath channels, we should note that
LoneSTAR could be extended with time-domain filtering or covariance-based modeling to capture such effects.

Overall, \gls{cissir} performs better than LoneSTAR in terms of the \gls{maxsi}.
In order to demonstrate once again this metric's relevance for \gls{isac}
under hybrid beamforming,
we compute range-angle maps with $L=K=4$ orthogonal beams in \cref{fig:rng-ang},
and we compare LoneSTAR and \gls{cissir} codebooks for
$\stx^2\approx$\SI{-18.2}{\deci\bel} and a single-path \gls{si} channel.
Here, we add a second target at \SI{30}{\meter} and \SI{20}{\degree}, and 
we cancel the residual \gls{si} digitally~\cite{alexandropoulos2022fullduplex} after \gls{adc} sampling.
Despite this digital \gls{si} cancellation,
we observe a higher noise floor for LoneSTAR in \cref{fig:ra-lst} than
for \gls{cissir} in \cref{fig:ra-csr}, as a result of the greater
 $\mSo(\C,\W)$ and the \gls{adc}'s limited dynamic range.

 \begin{table*}[!t]
\centering
\caption{Algorithmic metrics for LoneSTAR (phased and tapered) and \gls{cissir}}
\begin{tabular}{l|ccc|cc|cc}\label{tab:algo}
\textbf{\gls{si} channel model }         & \multicolumn{3}{c|}{\textbf{Optimization time (1st -- 3rd quartiles)}} & \multicolumn{2}{c|}{\textbf{Minimum feasible \gls{si}}}                    & \multicolumn{2}{c}{\textbf{Largest \gls{si} gap}}               \\ \hline
                        & LoneSTAR& Phased \gls{cissir}& Tapered \gls{cissir}& Phased \gls{cissir}& Tapered \gls{cissir}& Phased \gls{cissir}& Tapered \gls{cissir}\\
Single-path $\S_1(t)$ & \SIrange{0.16}{1.02}{\second}           & \SIrange{1.20}{1.48}{\second} & \SIrange{6.5}{7.7}{\milli\second} & \SI{-123}{\deci\bel} & \SI{-165.4}{\deci\bel}              & \SI{0.0}{\deci\bel}  & \SI{0.5}{\deci\bel}\\
Multipath $\S(t)$    & -                            & \SIrange{1.14}{1.39}{\second} & \SIrange{4.6}{8.2}{\milli\second} & \SI{-108}{\deci\bel} & \SI{-132.1}{\deci\bel} & \SI{0.4}{\deci\bel} & \SI{1.9}{\deci\bel}\end{tabular}
\end{table*}

\subsection{Algorithmic performance}
\label{sec:algo}

We finalize our investigation with some empirical results
on the time complexity and feasibility of \gls{cissir},
which we summarize in \cref{tab:algo}.

To analyze time complexity, we
report the first and third quartiles of the studied algorithms'
runtime on an AMD Epyc 73F3 processor.
Here, we can see that phased \gls{cissir}
is slower than LoneSTAR,
which can be explained by
the space lifting of the \gls{sdp} problem;
e.g.,
\cref{alg:sdr} must optimize
$L'$ $N\times N$ matrices
to retrieve the $N\times L'$-sized \gls{tx} codebook.
While slower than \SI{1}{\second},
the runtime of
\cref{alg:sdr} should represent
a negligible overhead
in a practical calibration scheme,
since the coherence time of
\gls{si} channels has been reported to be
in the order of minutes,
at least~\cite{roberts2023lonestar,roberts2022beamformed}.
On the other hand, 
tapered \gls{cissir} is almost 200 times faster than its phased counterpart
and up to 130 times faster than LoneSTAR,
thanks to the semi-closed-form solution of \cref{alg:lagrange}.

Regarding problem feasibility, we can
compute the minimum feasible \gls{si} exactly
for tapered \gls{cissir}
as $\min_{k,\ell}(\bar{\sigma}(\b_k,\Gr)\bar{\sigma}(\a_\ell,\Gt))$
(cf. \cref{pr:taper-sol}).
For phased arrays,
we estimate the minimum feasible \gls{si}
as the target \gls{si}
$\varepsilon$ for which
$\Upupsilon(\tilde{\Z})$ in \eqref{eq:rank1metric}
is above 0.995 on average.
In any case,
we conclude that \gls{cissir}
provides optimal solutions
in our scenarios for a wide range of
$\varepsilon$, over
\SI{50}{\deci\bel} below the reference \gls{maxsi}
in~\cref{tab:sysparams}.

Finally, we note that the original
\gls{si} constraint
\eqref{eq:eps-const}
is fulfilled quite tightly
in our simulations.
In fact, the ``\gls{si} gap'' between \gls{maxsi}
$\mS(\C,\W)$ and target \gls{si} $\varepsilon$ is at most
\SI{0.4}{\deci\bel} for \gls{maxsi} above \SI{-108}{\deci\bel} in our
experiments, and it only increases to \SI{2}{\deci\bel}
for multipath tapered \gls{cissir} when it approaches its feasible \gls{si}.
For single-path \gls{si}, the observed gap is at most
\SI{0.5}{\deci\bel}, which seems to confirm
the expected tightness of \eqref{eq:split-ineq} for
spatially sparse symmetric \gls{si} channels.
Furthermore, we have noticed similar gaps
for the saturation bound \eqref{eq:sat}
in the context of \cref{pr:sat}.
In conclusion, it can be argued that
\gls{cissir} is a valid approximation to the original
\cref{prob:joint},  based on the obtained \gls{si} gap values
from \cref{tab:algo}.
\section{Conclusion}
\label{sec:conc}

In this paper, we have proposed
\acf{cissir} as an evolution of the existing beam codebooks
in current mobile networks, both for hybrid and analog schemes.
\gls{cissir} achieves notable \acf{si} reduction and thus paves the way for \acf{isac} beyond 5G.

The proposed methodology allows a systematic optimization of
phased-array and tapered beam codebooks via
\acl{sdp} and spectral decomposition, respectively.
Our framework also provides performance guarantees
regarding multipath \gls{si} impairments in the analog domain.
Most notably, we have derived an analytical bound on the quantization noise
at the \acl{adc}, and we have proved its practicality via sensing simulations.
Moreover, our link-level results show how \gls{cissir} can reduce
\gls{si} with a small impact on 5G-\acs{nr}-based communication.
Contrary to other available \gls{si}-reduced beam codebooks,
our proposed method has easily interpretable optimization parameters,
and it outperforms said codebooks for relevant
\gls{isac} performance indicators.

Future research can introduce other practical aspects into our optimization framework,
such as subarrays, \gls{si} channel estimation errors, or quantized beamforming.
In doing so, we aim towards the hardware realization of a high dynamic-range
5G-\acs{nr} \gls{isac} proof of concept.

\appendix
\section{}
\crefalias{subsection}{appendix}

\subsection{Proof of
\texorpdfstring{\cref{prop:paralel}}{Proposition 1}
(\texorpdfstring{\acs{maxsi}}{max. SI} gap)}
\label{apx:split}

To prove $\mS(\C,\W)\leq\mgg(\C,\W)$
in \cref{prop:paralel},
we use the \acs{svd} in \cref{def:svd}
to bound the
argument inside $\mS(\C,\W)\coloneqq\max_{ k\in[K'],\ell\in[L']}\int_0^{\Tc}\norm{\c_k^\h\S(t)\w_\ell}dt$
as:
\begin{align*}
    &\int_0^{\Tc}\norm{\c_k^\h\S(t)\w_\ell}dt=
    \int_0^{\Tc}\norm{\c_k^\h
   \U_t\Eig_t^{1/2}\Eig_t^{1/2}\V^\h_t\w_\ell}dt\\
   \leq&
    \int_0^{\Tc}\Norm[2]{
    \Eig_t^{1/2}\U^\h_t\c_k}\Norm[2]{\Eig_t^{1/2}\V^\h_t\w_\ell}dt\\
    \leq&\sqrt{\int_0^{\Tc}\Norm[2]{\Eig_t^{1/2}\U^\h_t\c_k}^2\,dt
    \int_0^{\Tc}\Norm[2]{\Eig_t^{1/2}\V^\h_t\w_\ell}^2dt}\\
    =&\sqrt{\int_0^{\Tc}\c_k^\h\U_t\Eig_t\U_t^\h\c_k\,dt
    \int_0^{\Tc}\w_\ell^\h
    \V_t\Eig_t\V_t^\h\w_\ell\,dt}\,,
\end{align*}
where we have applied the Cauchy-Schwarz inequality twice,
first in its vector and then in its integral form.
The proof is completed by applying $\max$ over $k\in[K']$ and $\ell\in[L']$.

\subsection{Proof of \texorpdfstring{\cref{pr:taper-sol}}{Proposition 2}
(tapered beamforming solution)}
\label{apx:taper}

We proof the optimality of $\zsol\in\CC^P$ in
\cref{pr:taper-sol} by constructing
$\zsol$ such that it fulfills, together with
certain dual variables $\mu,\lambda>0$, the
\gls{kkt} conditions for the following
convex relaxation of \cref{prob:taper}
\begin{equation}\label[problem]{prob:tap-cvx}
    \underset{\z\in\CC^{P}}{\maximize}{\;\Re(\r^\h\z)}
    \quad\quad
    \st\quad\quad
    \z^\h\G\z\leq\epsilon,\quad\z^\h\z\leq 1\,.
\end{equation}
We note that \cref{prob:tap-cvx} fulfills Slater's condition,
as $\Zero$ is an interior point of the problem's feasible region.
Consequently, strong duality holds and thus any point
satisfying the \gls{kkt} conditions is optimal.

First, we recall the \gls{kkt} stationarity condition as
$\nabla\setL(\zsol,\mu,\lambda)=\Zero$
for the Lagrangian of \eqref{prob:tap-cvx}:
\begin{equation*}
   \setL(\z,\mu, \lambda)={-\Re(\r^\h\z)+
    \mu\left(\z^\h\G\z-\epsilon\right)+
    \lambda\left(\z^\h\z-1\right)},
\end{equation*}
from which it follows that\footnote{Formally, \eqref{eq:statty}
follows from the stationarity condition for a reformulation of
\cref{prob:tap-cvx} with real optimization variables, cf. \cite[Equation (2)]{fuchs2014application}.}
\begin{align}
    \nonumber
    -\r+2\mu\G\zsol+2\lambda\zsol=\Zero\;&\Rightarrow\;
    \zsol=\frac{1}{2}\left(\mu\G+\lambda\I\right)^{-1}\r=\\
    \eta\left(\Q\Gam\Q^\h+\nu\I\right)^{-1}\r&=\eta\Q\underbrace{\left(\Gam+\nu\I\right)^{-1}}_{\eqqcolon\D_\nu}\Q^\h\r
    \,,\label{eq:statty}
\end{align}
where we have applied the spectral decomposition of $\G=\Q\Gam\Q^\h$ and the variable changes $\eta=1/2\mu$
and $\nu=\lambda/\mu$ for $\mu,\lambda,\eta,\nu>0$.
Notice that $\zsol$ is obtained
by multiplying $\r$ with the positive definite matrix $\Q\D_\nu\Q^\h$, and scaling by $\eta$
to ensure $\Norm[2]{\zsol}=1$.

Since we are searching for strictly positive
dual variables,
$\zsol$ should satisfy all constraints in \eqref{prob:taper} and
\eqref{prob:tap-cvx} with equality
to fulfill complementary slackness, i.e.:
\begin{align}
    \nonumber
    1=\zsol^\h\zsol&=\eta\r^\h\Q\D_\nu\Q^\h\Q\D_\nu\Q^\h\r \eta=
    \eta^2\r^\h\Q\D_\nu^2\Q^\h\r\,,\\\nonumber\epsilon=\zsol^\h\G\zsol&=\eta\r^\h\Q\D_\nu\Q^\h\Q\Gam\Q^\h\Q\D_\nu\Q^\h\r \eta\\
    &=
    \eta^2\r^\h\Q\D_\nu\Gam\D_\nu\Q^\h\r\,,\label{eq:primal-feas}
\end{align}
and if we divide both expressions by $\eta^2>0$, we obtain:
\begin{align}\nonumber
    \r^\h\Q\D_\nu\Gam\D_\nu\Q^\h\r=&\frac{\epsilon}{\eta^2}=
    \epsilon\,\r^\h\Q\D_\nu^2\Q^\h\r\\
    \sum_{p=1}^P\frac{\sigma_p\norm{\q_p^\h\r}^2}{\left(\sigma_p+\nu\right)^2}&=
    \epsilon\sum_{p=1}^P\frac{\norm{\q_p^\h\r}^2}{\left(\sigma_p+\nu\right)^2}\,,
    \label{eq:lagrange-const-equal}
\end{align}
where we have used the diagonality of $\D_\nu$ and $\Gam$
to write out the matrix multiplications. We can now multiply both sides
of \eqref{eq:lagrange-const-equal}
by $\prod_{p=1}^P{\left(\sigma_p+\nu\right)^2}>0$,
and rearrange to obtain $\setP(\nu)=0$.
Hence, if $\setP(\nu)$ has a positive root $\nu_\star$,
the \gls{kkt} conditions of \eqref{prob:tap-cvx}
are fulfilled for $\zsol$, as given in \cref{pr:taper-sol}, with
\begin{equation*}
    \mu=\frac{1}{2\eta}=\frac{1}{2}\sqrt{\r^\h\Q\D_{\nu_\star}^2\Q^\h\r}>0\,,
    \quad\text{and}\quad
    \lambda=\nu_\star\mu>0\,.
\end{equation*}
More specifically,
($\mu,\lambda$) are dual feasible,
stationarity is fulfilled by \eqref{eq:statty},
and \eqref{eq:primal-feas} ensures both
complementary slackness and primal feasibility.
Moreover, $\zsol$ is also feasible for \eqref{prob:taper} by \eqref{eq:primal-feas},
and thus \eqref{prob:tap-cvx} becomes a tight convex relaxation of \eqref{prob:taper}.
The following \cref{thm:poly} completes the proof by guaranteeing the
existence of $\nu_\star>0$ under the conditions of \cref{pr:taper-sol}.

\begin{lemma}\label{thm:poly}
    Consider $\G=\Q\Gam\Q^\h\in\SS^P_+$,
    $\lbrace\sigma_p\rbrace_{p=1}^P$,
    $\r\in\setZ$, and $\epsilon>0$
    as defined in \cref{pr:taper-sol}. Any of the conditions \ref{cond:rank-def})
    and \ref{cond:rank-full}) therein are sufficient for the existence of a
    positive root for
\begin{equation*}
    \setP(\nu)\coloneqq\sum_{p=1}^P{\left(\sigma_p-\epsilon\right)\norm{\q_p^\h\r}^2}\prod_{q\neq p}{\left(\sigma_q+\nu\right)^2}\,.
\end{equation*}
\end{lemma}
\begin{proof}
    From Descartes' rule of signs~\cite{wang2004simple},
    the existence of a positive root $\nu_\star>0$ for $\setP(\nu)$
    is given by the sufficient but not necessary condition
    that $\setP(\nu)$ has nonzero highest- and lowest-degree
    coefficients with opposite signs.
    
    Under either condition \ref{cond:rank-def})
    or \ref{cond:rank-full}), we have $\epsilon<\r^\h\G\r$,
    and thus the highest-degree coefficient $c_{2P-2}$ is positive:
    \begin{align*}
        c_{2P-2}&=\sum_{p=1}^P{\left(\sigma_p-\epsilon\right)\norm{\q_p^\h\r}^2}\\
        &=
        \r^\h\sum_{p=1}^P\sigma_p\q_p\q_p^\h\r-\epsilon\r^\h\Q\Q^\h\r\\
        &=\r^\h\G\r-\epsilon\r^\h\r
        =\r^\h\G\r-\epsilon>0\,.
    \end{align*}
    For the lowest-degree coefficient, we need to consider the rank of $\G$.
    For the rank-deficient
    $\G$ in \ref{cond:rank-def}), it can be proven 
    by mathematical evaluation
    that the lowest nonzero coefficient has degree $2R-2$,
    with $R=P-\rank(\G)>0$ and $R<P$, and that it equals
    \begin{equation*}
         c_{2R-2}=-\epsilon\sum_{{p:\,\sigma_p=0}}\norm{\q_p^\h\r}^2
         \prod_{\sigma_q\neq0}{\sigma_q^2}<0\,.
    \end{equation*}
    Therefore, $\setP(\nu)$ always has a positive root
    $\nu_\star$
    for $\rank(\G)\in[1,\,P-1]$
    and $\epsilon\in(0,\r^\h\G\r)$, as given by \cref{cond:rank-def})
    in \cref{pr:taper-sol}.

    If $\G$ is full rank,
    then ($\forall\,p\in[P]$)
    $\sigma_p>0$, and the
    existence of $\nu_\star>0$ is guaranteed by $\epsilon<\r^\h\G\r$ and
    \begin{align*}
        c_0=\sum_{p=1}^P{\left(\sigma_p-\epsilon\right)\norm{\q_p^\h\r}^2}\prod_{q\neq p}{\sigma_q^2}&<0\\
        \sum_{p=1}^P{\frac{\left(\sigma_p-\epsilon\right)\norm{\q_p^\h\r}^2}{\sigma_p^2}}&<0\\
        \bar{\sigma}(\r,\G)=\frac{\sum_{p=1}^P{{\norm{\q_p^\h\r}^2}/{\sigma_p}}}
        {\sum_{p=1}^P{{\norm{\q_p^\h\r}^2}/{\sigma_p^2}}}&<\epsilon\,,
    \end{align*}
    as stated in condition \ref{cond:rank-full}).
\end{proof}

\subsection{Proof of \texorpdfstring{\cref{thm:quant-bound}}{Lemma 1}
(quantization noise bound)}
\label{apx:bound}

The total noise variance for $\yH[i]$ in \eqref{eq:noise-bound} can be rewritten as
\begin{multline}\label{eq:total-noise}
    \EE\left[\Norm[2]{\yH[i]-\y(iT_\mathrm{S})}^2\right]=
    \EE\left[\sum_{k=1}^K\norm{\yh_k[i]-{y}_k(iT_\mathrm{S})}^2\right]\\
    =\sum_{k=1}^K\EE[\norm{v_{k}[i]
    +u_{k}[i]}^2]\overset{\textrm{a)}}{=}
    \sum_{k=1}^K\EE[\norm{v_{k}[i]}^2]
    +\EE[\norm{u_{k}[i]}^2]\\
    =\sum_{k=1}^K\left(\sigma_\mathrm{T}^2+
    \sigma_\mathrm{Q}^2\right)
    \overset{\textrm{b)}}{=}
    K\sigma_\mathrm{T}^2+
    {K\bQ\yfs^2}\,,
\end{multline}
where we have used
a) the independence of $u_k[i]$ and $v_k[i]$, and
b) the definition of 
$ \sigma_\mathrm{Q}^2$ in
\eqref{eq:qnoise}.
Now, we turn our attention to $\yfs$ in  \eqref{eq:yfs}.
Using the maximum norm $\infnorm{\cdot}$,
we note that
\begin{align}
    \nonumber
    \yfs&=\gamma\max_{k\in[K]}
    \infnorm{\sum_{\ell=1}^L\left(\c_{\imath(k)}^\h\S\w_{\jmath(\ell)}\right)\ast x_\ell}\\&\leq
    \gamma\max_{k'\in[K']}
    \infnorm{\sum_{\ell=1}^L\left(\c_{k'}^\h\S\w_{\jmath(\ell)}\right)\ast x_\ell},
    \label{eq:yfs-k}
\end{align}
where we upper bound $\yfs$ independently of $\imath$ by considering
$k'$ in the full codomain of $\imath$.
For any $k'\in[K']$, we have:
\begin{multline}
    \label{eq:max_y_k}
    \infnorm{\sum_{\ell=1}^L\left(\c_{k'}^\h\S\w_{\jmath(\ell)}\right)\ast x_\ell}
   \overset{\textrm{a)}}{\leq}
   \sum_{\ell=1}^L\infnorm{\left(\c_{k'}^\h\S\w_{\jmath(\ell)}\right)\ast x_\ell}
   \overset{\textrm{b)}}{\leq}\\
   \sum_{\ell=1}^L\norm[1]{\c_{k'}^\h\S\w_{\jmath(\ell)}}\infnorm{x_\ell}
   \leq
   \max_{\ell'\in[L']}\norm[1]{\c_{k'}^\h\S\w_{\ell'}}\sum_{\ell=1}^L\infnorm{x_\ell}\,.
\end{multline}
In \eqref{eq:max_y_k}, a) follows from the triangular inequality, and b) from Young's convolution inequality.
From \eqref{eq:yfs-k} and \eqref{eq:max_y_k}, it follows that
\begin{equation}
\begin{alignedat}{2}\label{eq:yfs-bound}
    \yfs^2\leq
    &\gamma^2\max_{k'\in[K'],\ell'\in[L']}\norm[1]{\c_{k'}^\h\S\w_{\ell'}}^2\left(\sum_{\ell=1}^L\infnorm{x_\ell}\right)^2
    &\overset{\textrm{a)}}{=}&\\&\gamma^2\mS^2(\C,\W)\left(\sum_{\ell=1}^L\infnorm{x_\ell}\right)^2
    &\overset{\textrm{b)}}{=}&\\&
   \gamma^2\mS^2(\C,\W)\left(\sum_{\ell=1}^L\norm[2]{x_\ell}
   \sqrt{\rho_{x_\ell}/T}\right)^2
   &\overset{\textrm{c)}}{\leq}&\\&
   \gamma^2\mS^2(\C,\W)
   \sum_{\ell=1}^L\norm[2]{x_\ell}^2
   \sum_{\ell=1}^L\rho_{x_\ell}/T
   &\overset{\textrm{d)}}{=}&\\&
    \gamma^2\mS^2(\C,\W)
   \Pt
   \sum_{\ell=1}^L\rho_{x_\ell}\,.&&
\end{alignedat}
\end{equation}
Here we have applied
a) the definition of the \gls{maxsi} in \eqref{eq:maxsi},
b) the \gls{papr} expression \eqref{eq:papr},
c) the Cauchy-Schwarz inequality, and
d) the power relation \eqref{eq:Pt}.
The proof is completed by
upper bounding $\yfs^2$ in \eqref{eq:total-noise}
with the right-hand side of \eqref{eq:yfs-bound} and rearranging.

\section*{Acknowledgment}

The authors of this work acknowledge the financial support by
Germany's Federal Ministry for Research, Technology and Space (BMFTR)
in the programme ``Souverän. Digital. Vernetzt.''
Joint project 6G-RIC (grant numbers: 16KISK020K, 16KISK030).
Rodrigo Hernang\'{o}mez acknowledges BMFTR support in the project
``KOMSENS-6G'' (grant number: 16KISK121).


\ifCLASSOPTIONcaptionsoff
  \newpage
\fi



%
\bibliographystyle{IEEEtran}
\bibliography{IEEEabrv,references}


\end{document}